\documentclass[11pt]{article}

\usepackage[ruled]{algorithm2e}

\SetAlFnt{\small}

\usepackage{algorithmic}
\usepackage{subfigure}
\usepackage{amsthm}
\usepackage{amsmath}
\usepackage{amssymb}
\usepackage{amsfonts}
\usepackage{epsfig}
\usepackage{color}

\usepackage{framed}
\usepackage{psfrag}
\usepackage{fullpage}

\newtheorem{theorem}{Theorem}[section]
\newtheorem{definition}[theorem]{Definition}
\newtheorem{corollary}[theorem]{Corollary}
\newtheorem{proposition}[theorem]{Proposition}
\newtheorem{lemma}[theorem]{Lemma}

\newtheorem{observation}[theorem]{Observation}
\newtheorem{claim}[theorem]{Claim}
\def\m{{\rm min}}
\def\polylog{\operatorname{polylog}}
\renewcommand{\Pr}{\mathbb{P}}

\newcommand{\squishlist}{
 \begin{itemize}
}
\newcommand{\squishend}{
  \end{itemize}  }

\def\e{{\rm E}}

\def\bone{{\bf 1}}

\def\prasad#1{}
\def\danupon#1{}
\def\gopal#1{}
\def\atish#1{}

\begin{document}

\markboth{A. Das Sarma, D. Nanongkai, G. Pandurangan, P. Tetali}{Distributed Random Walks}


\title{Distributed Random Walks\thanks{Preliminary versions of this paper appeared in 28th ACM Symposium on Principles of Distributed Computing (PODC) 2009, Calgary, Canada and 29th ACM Symposium on Principles of Distributed Computing (PODC) 2010, Zurich, Switzerland~\cite{DNP09-podc,DNPT10-podc}.}   
%
%
}

\date{}

\author{Atish {Das Sarma}\thanks{eBay Research Labs, San Jose, CA, USA. \hbox{E-mail}:~{\tt atish.dassarma@gmail.com}. Work partially done while at  Georgia Institute of Technology and Google Research.}
\and Danupon Nanongkai\thanks{Division of Mathematical Sciences, Nanyang Technological University, Singapore 637371. \hbox{E-mail}:~{\tt danupon@gmail.com}. Work partially done while at Georgia Institute of Technology and  University of Vienna.}
\and Gopal Pandurangan\thanks{Division of Mathematical Sciences, Nanyang Technological University, Singapore 637371 and Department of Computer Science, Brown University, Providence, RI 02912, USA. \hbox{E-mail}:~{\tt gopalpandurangan@gmail.com}. Supported by the following grants: Nanyang Technological University grant M58110000, Singapore Ministry of Education (MOE) Academic Research Fund (AcRF) Tier 2 grant MOE2010-T2-2-082, US NSF grant CCF-1023166, and a grant from the US-Israeli Binational Science Foundation (BSF).}
\and Prasad Tetali\thanks{School of Mathematics and School of Computer Science, Georgia Institute of Technology Atlanta, GA 30332, USA. \hbox{E-mail}:~{\tt tetali@math.gatech.edu}. Supported in part by NSF DMS 0701023 and NSF CCR 0910584.}
}

\maketitle

\begin{abstract}
Performing random walks in networks is a fundamental primitive that has found applications in many areas of computer science, including distributed computing. In this paper, we focus on the problem of sampling random walks efficiently in a distributed network and its applications. Given bandwidth constraints, the goal is to minimize the number of rounds required to obtain random walk samples.

All previous algorithms that compute a random walk sample of length $\ell$ as a subroutine always do so naively, i.e., in $O(\ell)$ rounds. The main contribution of this paper is a fast distributed algorithm for performing random walks. We  present a sublinear time distributed algorithm for performing random walks whose time complexity is sublinear in the length of the walk. Our algorithm performs a random walk of length $\ell$  in $\tilde{O}(\sqrt{\ell D})$  rounds ($\tilde{O}$ hides $\polylog{n}$ factors where $n$ is the number of nodes in the network) with high probability on an undirected  network, where $D$ is the diameter of the network. For small diameter graphs, this is a significant improvement over the naive $O(\ell)$ bound. Furthermore,  our algorithm is optimal within a poly-logarithmic factor as there exists a matching  lower bound \cite{NanongkaiDP11}. We further extend our algorithms to efficiently perform $k$ independent random walks in $\tilde{O}(\sqrt{k\ell D} + k)$ rounds. We also show that our algorithm  can be applied to speedup the more general Metropolis-Hastings sampling.

Our random walk algorithms can be used to speed up distributed algorithms in applications that use random walks as a subroutine. We present two main applications. First, we give a fast distributed algorithm for computing a random spanning tree (RST) in an arbitrary (undirected unweighted) network which runs in $\tilde{O}(\sqrt{m}D)$ rounds with high probability ($m$ is the number of edges). Our second application is a fast decentralized algorithm for estimating mixing time and related parameters of the underlying network. Our algorithm is fully decentralized and can serve as a building block in the design of topologically-aware networks.
\end{abstract}

\section{Introduction}
\label{sec:intro}
Random walks play a central role in computer science,
spanning a
wide range of areas in both theory and practice. The focus of this
paper is on random walks in networks, in particular, decentralized
algorithms for performing random walks in arbitrary networks.
Random walks are used as an integral subroutine in a wide variety of
network applications ranging from token management~\cite{IJ90,BBF04,CTW93}, load balancing~\cite{KR04}, small-world
routing~\cite{K00}, search~\cite{ZS06,AHLP01,C05,GMS05,LCCLS02},
information propagation and gathering~\cite{BAS04,KKD01}, network
topology construction~\cite{GMS05,LawS03,LKRG03}, checking
expansion~\cite{DolevT10}, constructing random spanning
trees~\cite{Broder89,bar-ilan,Baala}, monitoring
overlays~\cite{MG07}, group communication in ad-hoc
network~\cite{DSW06}, gathering and dissemination of information
over a network \cite{aleliunas}, distributed construction of expander
networks \cite{LawS03}, and peer-to-peer membership
management~\cite{GKM03,ZSS05}.
Random walks  are also very useful in providing uniform and
efficient solutions to distributed control of dynamic networks
\cite{BBSB04,ZS06}.  Random walks are local and lightweight; moreover, they
require little index or state maintenance which makes them especially
attractive to self-organizing dynamic networks such as Internet
overlay and ad hoc wireless networks.

A key purpose of random walks in  many of these network applications
is to perform  node sampling.  While the sampling requirements in different
applications vary, whenever a true sample is required from a random
walk of certain steps, typically all applications perform the walk naively
--- by simply passing a token from one node to its neighbor: thus to
perform a random walk of length $\ell$ takes time linear in $\ell$.

In this paper, we  present an optimal (within a poly-logarithmic factor)  sublinear  time (sublinear in $\ell$) distributed  random walk
sampling algorithm that is significantly faster than the naive algorithm when $\ell \gg D$. Our algorithm runs in time $\tilde{O}(\sqrt{\ell D})$ rounds.
This running time is optimal (within a poly-logarithmic factor) since a matching lower bound was shown  recently in \cite{NanongkaiDP11}.  We then present two key applications of
our algorithm. The first is a fast distributed algorithm for
computing a random spanning tree, a fundamental problem that has
been studied widely in the classical setting (see e.g.,
\cite{kelner-madry} and references therein) and in some special
cases in distributed settings~\cite{bar-ilan}. To the best of our
knowledge, our algorithm gives the fastest known running time in an
arbitrary network. The second is to devising efficient decentralized
algorithms for computing key global metrics of the underlying
network --- mixing time, spectral gap, and conductance. Such
algorithms can be useful building blocks in the design of {\em
topologically (self-)aware} networks, i.e., networks that can
monitor and regulate themselves in a decentralized fashion. For
example,  efficiently computing the mixing time or the spectral gap,
allows  the network to monitor connectivity and expansion properties
of the network.

\subsection{Distributed Computing}
Consider an undirected, unweighted, connected $n$-node graph $G =
(V, E)$.  The network is modeled by an undirected $n$-vertex graph, where vertices model the processors and  edges model the links between the processors. Suppose that every node (vertex) hosts a processor with
unbounded computational power, but with limited initial knowledge.
The processors   communicate  by exchanging
messages via the links (henceforth, edges).  The vertices  have limited global knowledge, in particular, each of them has its own local perspective of the network, which is confined to its immediate neighborhood. Specifically, assume that each node is associated with a distinct
identity number from the set $\{1, 2, \ldots , \operatorname{poly}(n)\}$. At the beginning
of the computation, each node $v$ accepts as input its own identity
number and the identity numbers of its neighbors in $G$. The node
may also accept some additional inputs as specified by the problem
at hand. The nodes are allowed to communicate through the edges of
the graph $G$. The communication is synchronous, and occurs in
discrete pulses, called {\em rounds}. In particular, all the nodes
wake up simultaneously at the beginning of round 1. For convenience, our algorithms  assume that
nodes always know the number of the current round (although this is not really needed --- cf. Section~\ref{sec:algorithm}).

We assume the ${\cal CONGEST}$ communication  model, a
widely used standard model to study distributed
algorithms~\cite{peleg}: a node $v$ can send an arbitrary
message of size at most $O(\log n)$ through an edge per time step.
(We note that if unbounded-size messages were allowed through every
edge in each time step, then the problems addressed here can be
trivially solved in $O(D)$ time by collecting all  information at
one node, solving the problem locally, and then broadcasting the
results back to all the nodes \cite{peleg}.) The design of efficient algorithms for
the ${\cal CONGEST}$ model has been the subject of an active area of research called (locality-sensitive)  {\em distributed computing} (see \cite{peleg} and references therein.)
It is  straightforward to generalize our results to a
${\cal CONGEST}(B)$ model, where $O(B)$ bits can be transmitted in a
single time step across an edge.

There are several measures of efficiency of distributed algorithms,
but we will concentrate on one of them, specifically, {\em the
running time}, that is, the number of rounds of distributed
communication. (Note that the computation that is performed by the
nodes locally is ``free'', i.e., it does not affect the number of rounds.)
 Many
fundamental network problems such as minimum spanning tree, shortest
paths, etc. have been addressed in this model (e.g., see
\cite{lynch,peleg,PK09}). In particular, there has been much
research into designing very fast distributed   approximation
algorithms (that are even faster at the cost of producing
sub-optimal solutions) for many of these  problems (see e.g.,
\cite{elkin-survey,dubhashi,khan-disc,khan-podc}).  Such algorithms
can be useful for large-scale resource-constrained and
dynamic networks where running time is crucial.

\subsection{Problems}

We consider the following basic  random walk
problem.

\paragraph{Computing One Random Walk where Destination Outputs Source}
We are given an arbitrary undirected, unweighted, and connected
$n$--node network $G = (V,E)$ and a source node $s \in V$. The goal
is to devise a distributed algorithm such that, in the end, some
node $v$ outputs the ID of $s$, where $v$ is a destination node
picked according to the probability that it is the destination of a
random walk of length $\ell$ starting at $s$. For brevity, this problem will henceforth be simply called
{\em Single Random Walk}.

For clarity, observe that the following naive algorithm solves the
above problem in $O(\ell)$ rounds: The walk of length $\ell$ is
performed by sending a token for $\ell$ steps, picking a random
neighbor in each step. Then, the destination node $v$ of this walk
outputs the ID of $s$.
Our goal is to perform such sampling with significantly less number
of rounds, i.e., in time that is sublinear in $\ell$.  On the other
hand, we note that it can take too much time (as much as
$\Theta(|E|+D)$ time) in the ${\cal CONGEST}$  model to collect all
the topological information at some node (and then computing
the walk locally).

We also consider the following variations and generalizations  of the Single Random Walk problem.
\begin{enumerate}
\item \textit{$k$ Random Walks, Destinations output Sources ($k$-RW-DoS)}: We have $k$ sources $s_1, s_2, ..., s_k$ (not necessarily distinct) and we want each of
$k$ destinations to output the ID of its corresponding source.

\item \textit{$k$ Random Walks, Sources output Destinations ($k$-RW-SoD)}: Same as
above but we want each source to output the ID of its corresponding destination.

\item \textit{$k$ Random Walks, Nodes know their Positions
($k$-RW-pos)}: Instead of outputting the ID of source or
destination, we want each node to know its position(s) in the random
walk. That is, for each $s_i$, if $v_1, v_2, ..., v_\ell$ (where
$v_1=s_i$) is the resultant random walk starting at $s_i$, we want each
node $v_j$ in the walk to know the number $j$ at the end of the
process.
\end{enumerate}

Throughout this paper, we assume the standard (simple) random walk:
in each step, an edge is taken from the current node $v$ with
probability $1/\deg(v)$ where $\deg(v)$ is the degree of
$v$. Our goal is to output a true  random sample from the
$\ell$-walk distribution starting from $s$.

%

\subsection{Motivation}
There are two key motivations for obtaining sublinear time bounds.
The first is that in many algorithmic applications, walks of length
significantly greater than the network diameter are needed. For
example, this is necessary in both the  applications   presented
later in the paper, namely distributed computation of a random
spanning tree (RST) and  computation of mixing time. In the RST
algorithm, we need to perform a random walk of expected length
$O(mD)$ (where $m$ is the number of edges in the network). In
decentralized computation of mixing time, we need to perform walks
of length at least the mixing time which can be significantly larger
than the diameter (e.g., in a random geometric graph model
\cite{MP}, a popular model for ad hoc networks, the mixing time can
be larger than the diameter by a factor of $\Omega(\sqrt{n})$.) More
generally, many real-world communication networks  (e.g., ad hoc
networks and peer-to-peer networks) have relatively small diameter,
and random walks of length at least the diameter are usually
performed for many sampling applications, i.e., $\ell \gg D$. It
should be noted that  if the network is rapidly mixing/expanding
which is sometimes the case in practice, then sampling from walks of
length $\ell \gg D$ is close to sampling from the steady state
(degree) distribution; this can be done in $O(D)$ rounds (note
however, that this gives only an approximately close sample, not the
exact sample for that length). However, such an approach fails when
$\ell$ is smaller than the mixing time.

The second motivation is understanding the time complexity of
distributed random walks. Random walk is essentially a ``global"
problem  which requires the algorithm to ``traverse" the entire
network. Classical global problems include the minimum spanning
tree, shortest path etc. Network diameter is an inherent lower bound
for such problems. Problems of this type raise the basic question
whether $n$ (or $\ell$ as is the case here) time is essential or is the
network diameter $D$, the inherent parameter. As pointed out in the
work of \cite{peleg-mst}, in the latter case, it would be
desirable to design algorithms that have a better complexity for
graphs with low diameter.

\medskip
\noindent \textbf{Notation:} Throughout the paper, we let $\ell$ be
the length of the walks, $k$ be the number of walks, $D$ be the
network diameter, $\delta$ be the minimum node degree, $n$ be the
number of nodes, and $m$ be the number of edges in the network.

\subsection{Our Results}


\paragraph{A Fast Distributed Random Walk Algorithm}
We present the first sublinear,  time-optimal, distributed
algorithm for the 1-RW-DoS problem in arbitrary networks
that runs in  time $\tilde{O}(\sqrt{\ell D})$ with high probability\footnote{Throughout this paper, ``with high probability (whp)" means
with probability at least $1 - 1/n^{\Omega(1)}$, where $n$ is the number of nodes in the network.}, where $\ell$ is the
length of the walk (the precise theorem is stated in Section \ref{sec:algorithm}). Our algorithm is randomized (Las Vegas type, i.e., it
always outputs the correct result, but the running time claimed is
with high probability).

The high-level idea behind our algorithm is to
``prepare'' a few short walks in the beginning and carefully stitch
these walks together later as necessary. If there are not enough
short walks, we construct more of them on the fly. We overcome a key
technical problem by showing how one can perform many short walks in
parallel without causing too much congestion.

Our algorithm exploits a certain key property of random walks. The key property is  a
bound on the number of times any node is visited in an $\ell$-length
walk, for any  length $\ell = O(m^2)$. We prove that w.h.p. any node
$x$ is visited at most $\tilde{O}(\deg(x)\sqrt{\ell})$ times, in an
$\ell$-length walk from any starting node ($\deg(x)$ is the degree of
$x$).  We then show that if only certain $\ell/\lambda$ special
points of the walk (called {\em connector points}) are observed,
then any node is observed only $\tilde{O}(\deg(x)\sqrt{\ell}/\lambda)$
times. The algorithm starts with all nodes performing short walks
(of length uniformly random in the range $\lambda$ to $2\lambda$ for
appropriately chosen $\lambda$) efficiently and simultaneously; here the
randomly chosen lengths play a crucial role in arguing about a
suitable spread of the connector points.   Subsequently, the
algorithm begins at the source and carefully stitches these walks
together till $\ell$ steps are completed.

We note that the running time of our algorithm matches the unconditional lower bound recently shown in \cite{NanongkaiDP11}.   Thus the running time of our algorithm is (essentially)  the best possible (up to polylogarithmic factors).

We also extend the result to give algorithms for computing $k$
random walks (from any $k$ sources
 ---not necessarily distinct) in $\tilde O\left(\min(\sqrt{k\ell D}+k, k+\ell)\right)$ rounds. We  note that the $k$ random walks generated by our algorithm are {\em independent} (cf. Section~\ref{subsec:many walks}). Computing $k$ random
walks is useful in many applications such as the one we present below on
decentralized computation of mixing time and related parameters. While the main requirement of our algorithms is to just obtain the random walk samples (i.e. the end point of the $\ell$ step walk), our algorithms can regenerate the entire walks such that each node knows its position(s) among the $\ell$ steps (the $k$-RW-pos problem).
Our algorithm
can  be extended to do this in the same number of rounds.

We finally present extensions of our algorithm to perform random
walk according to the Metropolis-Hastings~\cite{Hastings70,MRRT53}
algorithm, a more general type of random walk with numerous
applications (e.g., \cite{ZS06}). The Metropolis-Hastings  algorithm
gives a way to define  transition probabilities so that a random
walk converges to any desired distribution. An important special
case is when the distribution is uniform.

\paragraph{Remarks} While the message complexity is not the main focus of this paper, we note that our improved running time comes with the cost of an increased message complexity from the naive algorithm (we  discuss this in Section~\ref{sec:conclusion}). Our message complexity  for computing a random walk of length $\ell$  is $\tilde O(m\sqrt{\ell D}+n\sqrt{\ell/D})$ which can be worse than the naive algorithm's $\tilde O(\ell)$ message complexity.

%

\paragraph{Applications} Our faster distributed random walk algorithm
can be used in speeding up distributed applications where  random
walks arise as a subroutine. Such applications include distributed
construction of expander graphs, checking whether a graph is an
expander, construction of random spanning trees, and random-walk
based search (we refer to \cite{DNP09-podc} for details). Here, we
present two key applications:

(1) {\em A Fast Distributed Algorithm for Random Spanning Trees (RST):}
We give an $\tilde{O}(\sqrt{m}D)$ time distributed algorithm (cf. Section \ref{sec:rst}) for uniformly sampling a random spanning tree in an arbitrary undirected
(unweighted) graph (i.e., each spanning tree in the underlying network has the same probability of being selected). 
Spanning trees are fundamental network primitives
and distributed algorithms for various types of spanning trees such as minimum spanning tree (MST), breadth-first spanning tree (BFS), shortest path tree,
shallow-light trees etc., have been studied extensively in the literature \cite{peleg}. However, not much is known about the distributed complexity
of the random spanning tree problem.
The centralized case
has been studied for many decades, see e.g.,
the recent work of \cite{kelner-madry} and the references therein; also see the recent work of Goyal et al.
\cite{goyal} which gives nice applications of RST to fault-tolerant routing
and constructing expanders.  In the distributed computing context, the work
of Bar-Ilan and Zernik \cite{bar-ilan} give distributed RST algorithms for
two  special cases, namely that of a complete graph (running in constant time) and a synchronous ring (running in  $O(n)$ time).
The work of \cite{Baala}
gives a self-stablizing distributed algorithm for constructing an RST in a wireless ad hoc network and mentions that RST is more resilient to transient
failures that occur in mobile ad hoc networks.

Our algorithm works by giving an efficient distributed implementation
of the well-known Aldous-Broder random walk algorithm \cite{aldous,Broder89} for constructing an RST.

(2) {\em Decentralized Computation of Mixing Time.} We present a fast decentralized algorithm for estimating mixing time, conductance and spectral gap of the network (cf. Section~\ref{sec:mixingtime}). In
particular, we show that given a starting point $x$, the mixing time with respect to $x$, called $\tau^x_{mix}$, can be
estimated in $\tilde{O}(n^{1/2} + n^{1/4}\sqrt{D\tau^x_{mix}})$ rounds. This gives an alternative algorithm to the only previously known
approach by Kempe and McSherry \cite{kempe} that can be used to estimate
$\tau^x_{mix}$ in $\tilde O(\tau^x_{mix})$ rounds.\footnote{Note
that \cite{kempe} in fact does more and gives a decentralized algorithm for
computing the top $k$ eigenvectors of a weighted adjacency matrix
that runs in $O(\tau_{mix}\log^2 n)$ rounds if two adjacent nodes are allowed to exchange $O(k^3)$ messages per round, where $\tau_{mix}$ is
the mixing time and $n$ is the size of the network.}  To compare,  we note that
when $\tau^x_{mix} = \omega(n^{1/2})$ the present algorithm is faster (assuming
$D$ is not too large).


\subsection{Related Work}

Random walks have been used in a wide variety of applications in distributed networks as mentioned in the beginning of Section \ref{sec:intro}. We describe here some of the applications in more detail.
 Our focus is to emphasize the papers of a more theoretical nature, and those that use random walks as one of the central subroutines.

Speeding up distributed algorithms using random walks has been considered for a long time. Besides our approach of speeding up the random walk itself, one popular approach is to reduce the {\it cover time}. 
Recently, Alon et. al.~\cite{AAKKLT} show that performing several
random walks in parallel reduces the cover time in various types of
graphs. They assert that the problem with performing random walks is
often the latency. In these scenarios where many walks are
performed, our results could help avoid too much latency and yield
an additional speed-up factor. Other recent works involving multiple random walks in
different settings include  Els{\"a}sser
et. al.~\cite{ElsasserS09},  and Cooper et al. \cite{frieze}.

 A nice application of random walks is in  the design and analysis of expanders. We mention two results here. Law and Siu~\cite{LawS03} consider the problem of constructing expander graphs in a distributed fashion. One of the key subroutines in their algorithm is to perform several random walks from specified source nodes. While the overall running time of their algorithm depends on other factors, the specific step of
computing random walk samples can be improved using our techniques
presented in this paper. Dolev and Tzachar~\cite{DolevT10} use  random
walks to check if a given graph is an expander. The first algorithm
given in \cite{DolevT10} is essentially to run a random walk of length
$n\log{n}$ and mark every visited vertices. Later, it is checked if
every node is visited.

Broder~\cite{Broder89} and Wilson~\cite{Wilson96} gave algorithms to
generate random spanning trees using random walks and Broder's
algorithm was later applied to the network setting by Bar-Ilan and
Zernik~\cite{bar-ilan}. Recently Goyal et al.~\cite{goyal} show how
to construct an expander/sparsifier using random spanning trees. If
their algorithm is implemented on a distributed network, the
techniques presented in this paper would yield an additional
speed-up in the random walk constructions.


Morales and Gupta~\cite{MG07} discuss about discovering a consistent
and available monitoring overlay for a distributed system. For each
node, one needs to select and discover a list of nodes that would
monitor it. The monitoring set of nodes need to satisfy some
structural properties such as consistency, verifiability, load
balancing, and randomness, among others. This is where random walks
come in. Random walks is a natural way to discover a set of random
nodes that are spread out (and hence scalable), that can in turn be
used to monitor their local neighborhoods. Random walks have been
used for this purpose in another paper by Ganesh et al.~\cite{GKM03}
on peer-to-peer membership management for gossip-based protocols.

The general high-level idea  of using  a few short walks in the
beginning (executed in parallel) and then carefully stitch these
walks together later as necessary was introduced in~\cite{AtishGP08} to find random walks in data streams with the
main motivation of computing PageRank.
However, the two models have very different constraints and
motivations and hence the subsequent techniques used here and in
\cite{AtishGP08} are very different.
Recently, Sami and Twigg~\cite{ST08} consider lower bounds on the
communication complexity of computing the stationary distribution of
random walks in a network. Although their problem is related to
our problem, the lower bounds obtained do not  imply anything in our
setting.

The
work of \cite{mihail-topaware}  discusses spectral algorithms for
enhancing the topology awareness, e.g., by identifying and assigning
weights to critical links. However, the algorithms are centralized,
and it is mentioned that obtaining efficient decentralized
algorithms is a major open problem. Our algorithms are fully decentralized
and  based on performing random walks,
and so are more amenable to dynamic and self-organizing networks.

\paragraph{Subsequent Work} Since the publication of the conference versions
of our papers ~\cite{DNP09-podc,DNPT10-podc}, additional results have been shown,
extending our algorithms to various settings.

The work of  \cite{NanongkaiDP11} showed
a tight lower bound on the running time of distributed random walk algorithms using techniques from
communication complexity \cite{stoc11}. Specifically,  it is shown in \cite{NanongkaiDP11} that for any $n$, $D$, and $D\leq \ell \leq (n/(D^3\log n))^{1/4}$, performing a random walk of length $\Theta(\ell)$ on an $n$-node network of diameter $D$ requires $\Omega(\sqrt{\ell D}+D)$ time.  This shows that  the running time of our 1-RW-DoS algorithm is (essentially)  the best possible (up to polylogarithmic factors).

In \cite{infocom2012}, it is shown how to improve the message complexity of the distributed random
walk algorithms presented in this paper. The main reason for the increased message complexity of our algorithms
is that to compute one long walk many short walks are generated --- most of which go unused. One idea is
to use these unused short walks to compute other (independent) long walks. This idea is explored in \cite{infocom2012} where it is shown that   under certain conditions (e.g., when the starting point of the random walk is chosen
proportional to the node degree), the overall message complexity of computing many long walks can be made near-optimal.

The fast distributed random walk algorithms presented in this paper applies only for
 {\em static} networks and  does not apply to a dynamic network.
The recent work of \cite{disc2012} investigates efficient distributed computation in  dynamic networks in which the network topology changes (arbitrarily) from round to round.
The paper presents a rigorous framework for design and analysis of distributed random walk sampling algorithms in dynamic networks. Building on the techniques developed in the present paper, the main contribution of \cite{disc2012} is a fast distributed random walk
sampling algorithm that runs in $\tilde{O}(\sqrt{\tau \Phi})$ rounds (with high probability) ($\tau$ is the {\em dynamic mixing time}
and $\Phi$ is the {\em dynamic diameter} of the network) and  returns a sample close to a suitably defined stationary distribution of the dynamic network. This is then shown to be useful in designing
a fast distributed algorithm for information spreading in a dynamic network.

\section{Algorithm for {1-RW-D{\MakeLowercase o}S}} \label{sec:algorithm}

In this section we describe the algorithm to sample one random walk
destination. We show that this algorithm takes $\tilde O(\sqrt{\ell
D})$ rounds with high probability and extend it to other cases in
the next sections. First, we make the following simple observation,
which will be assumed throughout.

\begin{observation}\label{obs:length at most m2}
We may assume that $\ell$ is $O(m^2)$, where $m$ is the number of
edges in the network.
\end{observation}

The reason is that if $\ell$ is $\Omega(m^2)$, the required bound of
$\tilde O(\sqrt{\ell D})$ rounds is easily achieved by aggregating
the graph topology (via upcast) onto one node in $O(m+D)$ rounds
(e.g., see \cite{peleg}). The difficulty lies in proving the case of  $\ell
= O(m^2) $.

\paragraph{A Slower algorithm} Let us first consider a slower version of
the algorithm to highlight the fundamental idea used to achieve the
sub-linear time bound. We will show that the slower algorithm runs in time
$\tilde{O}(\ell^{2/3}D^{1/3})$.
The high-level idea (see Figure~\ref{fig:connector}) is to perform ``many" short random walks in
parallel and later stitch them together as needed. In particular, we
perform the algorithm in two phases, as follows.

In Phase~1, we perform $\eta$ ``short'' random walks of length
$\lambda$ from each node $v$, where $\eta$ and $\lambda$ are some
parameters whose values will be fixed in the analysis. (We note that we will need slightly more short walks when we develop a faster algorithm.)  This is done
naively by forwarding $\eta$ ``coupons'' having the ID of $v$, from
$v$ to random destinations\footnote{The term ``coupon'' refers to
the same meaning as the more commonly used term of ``token'' but we
use the term coupon here and reserve the term token for the second
phase.}, as follows.

\vspace{5pt}

\begin{algorithmic}[1]

\STATE Initially, each node $v$ creates $\eta$ messages (called
coupons) $C_1, C_2, ..., C_\eta$ and writes its ID on them.

\FOR{$i=1$ to $\lambda$}

\STATE This is the $i$-th iteration. Each node $v$ does the
following: Consider each coupon $C$ held by $v$ which is received in
the $(i-1)$-th iteration. (The zeroth iteration is the initial stage where each node creates its own messages.) Now $v$ picks a neighbor $u$ uniformly at
random and forwards $C$ to $u$ after incrementing the counter on the
coupon to $i$.


\ENDFOR
\end{algorithmic}

\vspace{5pt} At the end of the process, for each node $v$, there
will be $\eta$ coupons containing $v$'s ID distributed to some nodes
in the network. These nodes are the destinations of short walks of
length $\lambda$ starting at $v$.
We note that the notion of ``phase" is used only for simplicity.
The algorithm does not really need round numbers.
If there are many messages to be sent through the same edge, send one
with minimum counter first.

For Phase~2, for sake of exposition, let us first consider an easier version of the algorithm (that is  incomplete) which avoids some details. Starting at source
$s$, we ``stitch'' some of the $\lambda$-length walks prepared in Phase~1
together to form a longer walk. The algorithm starts from $s$ and
randomly picks one coupon distributed from $s$ in Phase~1. This can be
accomplished by having every node holding coupons of $s$ write their
IDs on the coupon and sending the coupons back to $s$. Then $s$ picks
one of these coupons randomly and returns the rest to the owners.
(However, aggregating all coupons at $s$ is inefficient. The better way to do this is to use the idea of {\em reservoir sampling}~\cite{Vitter85}.  We will develop an algorithm called {\sc Sample-Coupon} to do this job efficiently later on.)


%
Let $C$ be the sampled coupon and $v$ be the destination node of
$C$. The source $s$ then sends a ``token'' to $v$ and $v$ deletes coupon
$C$ (so that $C$ will not be sampled again next time). The process
then repeats. That is, the node $v$ currently holding the token
samples one of the coupons it distributed in Phase~1 and
forwards the token to the destination of the sampled coupon, say
$v'$. (Nodes $v$, $v'$ are called ``connectors" --- they are the endpoints of the short walks that are stitched.) A crucial observation is that the walk of length $\lambda$
used to distribute the corresponding coupons from $s$ to $v$ and
from $v$ to $v'$ are independent random walks. Therefore, we can
stitch them to get a random walk of length $2\lambda$. (This fact
will be formally proved in the next section.) We therefore can
generate a random walk of length $3\lambda, 4\lambda, ...$ by
repeating this process. We do this until we have completed more than
$\ell-\lambda$ steps. Then, we complete the rest of the walk by
running the naive random walk algorithm. The algorithm for Phase~2 is
thus the following.

\vspace{5pt}
\begin{algorithmic}[1]

\STATE The source node $s$ creates a message called ``token'' which
contains the ID of $s$

\WHILE {Length of the walk completed is at most $\ell-\lambda$}

  \STATE Let $v$ be the node that is currently holding the token.

  \STATE \label{line:before-send-more-coupons} $v$ calls {\sc Sample-Coupon($v$)} to sample
  one of the coupons distributed by $v$ (in Phase~1) uniformly at random. Let $C$
  be the sampled coupon.

  \STATE \label{line:after-send-more-coupons} Let $v'$ be the node holding coupon $C$. (ID of $v'$ is
  written on $C$.)

  \STATE $v$ sends the token to $v'$ and $v'$ deletes $C$ so that $C$ will not be sampled again.

  \STATE The length of the walk completed has now increased by $\lambda$.

\ENDWHILE

\STATE Walk naively (i.e., forward the token to a random neighbor)
until $\ell$ steps are completed.

\STATE A node holding the token outputs the ID of $s$.

\end{algorithmic}

\begin{figure}
\begin{center}
\includegraphics[width=\linewidth]{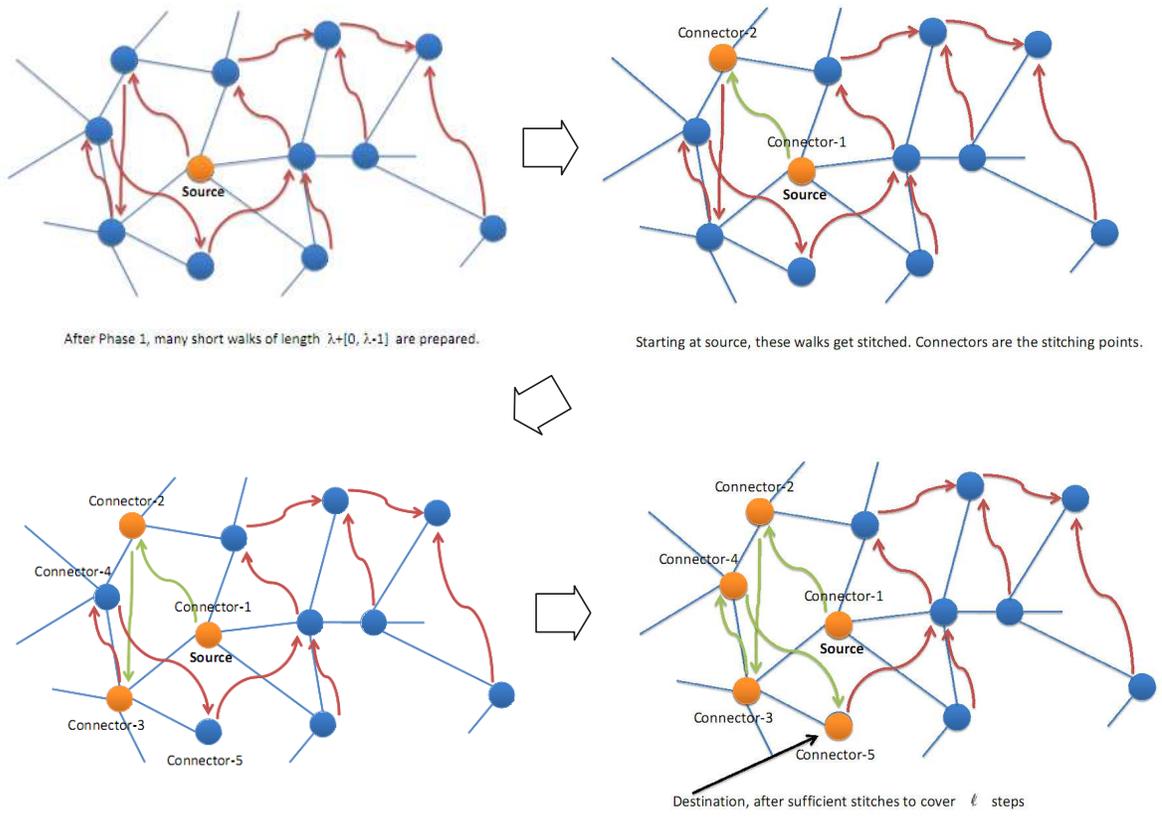}
\end{center}
\caption{Figure illustrating the algorithm of stitching short walks together.}
\label{fig:connector}
\end{figure}

\vspace{5pt}

Figure~\ref{fig:connector} illustrates the idea of this algorithm. To understand the intuition behind this (incomplete) algorithm, let us analyze its running time. First, we claim that Phase~1 needs $\tilde O(\eta\lambda)$ rounds with high probability. This is because if we send out $\deg(v)$ coupons from each node $v$ at the same time, each edge should receive  two coupons in the average case. In other words, there is essentially no congestion  (i.e., not too many coupons are sent through
the same edge). Therefore sending out (just) one coupon from each node for
$\lambda$ steps will take $O(\lambda)$ rounds in expectation and the
time becomes $O(\eta\lambda)$ for $\eta$ coupons. This argument can
be modified to show that we need $\tilde O(\eta\lambda)$ rounds with
high probability. (The full proof will be provided in
Lemma~\ref{lem:phase1} in the next section.)
We will also show that {\sc Sample-Coupon} can be done in $O(D)$
rounds and it follows that Phase~2 needs $O(D\cdot \ell/\lambda)$
rounds. Therefore, the algorithm needs $\tilde O(\eta\lambda+D\cdot
\ell/\lambda)$ which is $\tilde O(\sqrt{\ell D})$ when we set
$\eta=1$ and $\lambda=\sqrt{\ell D}$.

The reason the above algorithm for Phase~2 is incomplete is
that it is possible that $\eta$ coupons are not enough: We might
forward the token to some node $v$ many times in Phase~2 and all
coupons distributed by $v$ in the first phase may get deleted. In other words,
$v$ is chosen as a connector node many times, and all its coupons have been exhausted. If this happens then the stitching process cannot progress.
To cope with this problem, we develop an algorithm called {\sc
Send-More-Coupons} to distribute more coupons. In particular, when
there is no coupon of $v$ left in the network and $v$ wants to
sample a coupon, it calls {\sc Send-More-Coupons} to send out $\eta$
new coupons to random nodes. ({\sc Send-More-Coupons} gives the same
result as Phase~1 but the algorithm will be different in order to
get a good running time.) In particular, we insert the following
lines between Line~\ref{line:before-send-more-coupons} and
\ref{line:after-send-more-coupons} of the previous algorithm.

\vspace{5pt}
\begin{algorithmic}[1]

\IF {$C$ = {\sc null} (all coupons from $v$ have already been
deleted)}

  \STATE $v$ calls {\sc Send-More-Coupons($v$, $\eta$, $\lambda$)}
  (Distribute $\eta$ new coupons. These coupons are forwarded for
  $\lambda$ rounds.)

  \STATE $v$ calls {\sc Sample-Coupon($v$)} and let $C$ be the
  returned coupon.

  \ENDIF

\end{algorithmic}

\vspace{5pt} To complete this algorithm we now describe {\sc
Sample-Coupon} and {\sc Send-More-Coupons}.
The main idea of algorithm {\sc Sample-Coupon} is to sample the
coupons through a BFS (breadth-first search) tree from the leaves upward to the root. We
allow each node to send only one coupon to its parent to avoid
congestion. That is, in each round some node $u$ will receive some
coupons from its children (at most one from each child). Let these
children be $u_1, u_2, ..., u_q$. Then, $u$ picks one of these
coupons and sends to its parent. To ensure that $u$ picks a coupon
with uniform distribution, it picks the coupon received from $u_i$
with probability proportional to the number of coupons in the
subtree rooted at $u_i$. The precise statement of this algorithm can
be found in Algorithm~\ref{alg:Sample-Coupon}.
The correctness of this algorithm (i.e., it outputs a coupon from
uniform probability) will be proved in the next section (cf.
Claim~\ref{claim:correctness-sample-coupon}).

\begin{algorithm}[t]
\caption{\sc Sample-Coupon($v$)} \label{alg:Sample-Coupon}
\textbf{Input:} Starting node $v$.\\
\textbf{Output:} A node sampled from among the nodes holding the
coupon of $v$

\begin{algorithmic}[1]

\STATE Construct a Breadth-First-Search (BFS) tree rooted at $v$.
While constructing, every node stores its parent's ID. Denote such
a tree by $T$.

\STATE We divide $T$ naturally into levels $0$ through $D$ (where
nodes in level $D$ are leaf nodes and the root node $v$ is in level
$0$).


\STATE Every node $u$ that holds some coupons of $v$ picks one
coupon uniformly at random. Let $C_0$ denote such a coupon and let
$x_0$ denote the number of coupons $u$ has. Node $u$ writes its ID
on coupon $C_0$.

\FOR{$i=D$ down to $0$}

\STATE Every node $u$ in level $i$ that either receives coupon(s)
from children or possesses coupon(s) itself do the following.

\STATE Let $u$ have $q$ coupons (including its own coupons). Denote
these coupons by $C_0, C_1, C_2, \ldots, C_{q-1}$ and let their counts
be $x_0, x_1, x_2, \ldots, x_{q-1}$. Node $u$ samples one of $C_0$
through $C_{q-1}$, with probabilities proportional to the respective
counts. That is, for any $0\leq j\leq q-1$, $C_j$ is sampled with
probability $\frac{x_j}{x_0+x_1+\ldots+x_{q-1}}$.

\STATE The sampled coupon is sent to the parent node (unless already
at root) along with a count of $x_0+x_1+\ldots+x_{q-1}$ (the count
represents the number of coupons from which this coupon has been
sampled).

\ENDFOR

\STATE The root outputs the ID of the owner of the final sampled
coupon (written on such a coupon).

\end{algorithmic}

\end{algorithm}

The {\sc Send-More-Coupons} algorithm does essentially the same as
what we did in Phase~1 with only one exception: Since this time we
send out coupons from only one node, we can avoid congestions by
combining coupons delivered on the same edge in each round. This
algorithm is described in Algorithm~\ref{alg:Send-More-Coupons},
Part~1. (We will describe Part~2 later after we explain how to speed
up the algorithm).

\begin{algorithm}[t]
\caption{\sc Send-More-Coupons($v$, $\eta$, $\lambda$)}
\label{alg:Send-More-Coupons}

\paragraph{Part 1} Distribute $\eta$ {\em new} coupons for $\lambda$ steps.

\begin{algorithmic}[1]
\STATE The node $v$ constructs $\eta$ (identical) messages
containing its ID. We refer to these messages {\em new coupons}.

\FOR{$i=1$ to $\lambda$}

\STATE Each node $u$ does the following:

\STATE - For each new coupon $C$ held by $u$, node $u$ picks a neighbor $z$
uniformly at random as a receiver of $C$.

\STATE - For each neighbor $z$ of $u$, node $u$ sends the ID of $v$ and the number
of new coupons for which $z$ is picked as a receiver, denoted by $c(u, v)$.

\STATE - Each neighbor $z$ of $u$, upon receiving ID of $v$ and
$c(u, v)$, constructs $c(u, v)$ new coupons, each containing the ID of
$v$.

\ENDFOR

\end{algorithmic}

\paragraph{Part 2} Each coupon has now been forwarded for $\lambda$ steps. These
coupons are now extended probabilistically further by $r$ steps
where each $r$ is independent and uniform in the range
$[0,\lambda-1]$.

\begin{algorithmic}[1]

\FOR{$i=0$ to $\lambda-1$}

\STATE \label{line:reservoir} For each coupon, independently with
probability $\frac{1}{\lambda-i}$, stop sending the coupon further
and save the ID of the source node (in this event, the node with the
message is the destination). For each coupon that is not stopped,
each node picks a neighbor correspondingly and sends the coupon
forward as before.

\ENDFOR

\STATE At the end, each destination node knows the source ID as well
as the number of times the corresponding coupon has been forwarded.

\end{algorithmic}

\end{algorithm}

The analysis in the next section shows that {\sc Send-More-Coupons}
is called at most $\ell/(\eta\lambda)$ times in the worst case and
it follows that the algorithm above takes time $\tilde
O(\ell^{2/3}D^{1/3})$.

\begin{algorithm}
\caption{\sc Single-Random-Walk($s$, $\ell$)}
\label{alg:single-random-walk}

\textbf{Input:} Starting node $s$, desired walk length $\ell$ and parameters $\lambda$ and $\eta$.

\textbf{Output:} A destination node of the random walk of length $\ell$ output the ID of $s$.

\textbf{Phase 1: Generate short walks by coupon distribution.} Each
node $v$ performs $\eta\deg(v)$ random walks of length $\lambda +
r_i$ where $r_i$ (for each $1\leq i\leq \eta\deg(v)$) is chosen
independently and uniformly at random in the range $[0,\lambda-1]$. (We note that random numbers $r_i$ generated by different nodes are different.)
At the end of the process, there are $\eta\deg(v)$ (not necessarily
distinct) nodes holding a ``coupon'' containing the ID of $v$.

\begin{algorithmic}[1]

\FOR{each node $v$}

\STATE Generate $\eta\deg(v)$ random integers in the range $[0,
\lambda-1]$, denoted by $r_1, r_2, ..., r_{\eta\deg(v)}$.

\STATE Construct $\eta\deg(v)$ messages containing its ID and in
addition, the $i$-th message contains the desired walk length of
$\lambda + r_i$. We will refer to these messages created by node $v$
as ``coupons created by $v$''.

\ENDFOR

\FOR{$i=1$ to $2\lambda$}

\STATE This is the $i$-th iteration. Each node $v$ does the
following: Consider each coupon $C$ held by $v$ which is received in
the $(i-1)$-th iteration. (The zeroth iteration is the initial stage where each node creates its own messages.) If the coupon $C$'s desired walk length is
at most $i$, then $v$ keeps this coupon ($v$ is the desired
destination). Else, $v$ picks a neighbor $u$ uniformly at random and
forwards $C$ to $u$.

\ENDFOR
\end{algorithmic}

\textbf{Phase 2: Stitch short walks by token forwarding.} Stitch $\Theta(\ell/\lambda)$ walks, each of length in
$[\lambda,2\lambda-1]$.

\begin{algorithmic}[1]

\STATE The source node $s$ creates a message called ``token'' which
contains the ID of $s$

\STATE The algorithm will forward the token around and keep track of
a set of {\em connectors}, denoted by $\cal C$. Initially, ${\cal C}
= \{s\}$.

\WHILE {Length of the walk completed is at most $\ell-2\lambda$}

  \STATE Let $v$ be the node that is currently holding the token.

  \STATE $v$ calls {\sc Sample-Coupon($v$)} to uniformly sample one of the
  coupons distributed by $v$. Let $C$ be the sampled coupon.

  \IF{$v'$ = {\sc null} (all coupons from $v$ have already been deleted)}

  \STATE $v$ calls {\sc Send-More-Coupons($v$, $\eta$, $\lambda$)} (Perform $\Theta(\eta)$ walks
  of length $\lambda+r_i$ starting at $v$, where
$r_i$ is chosen uniformly at random in the range $[0,\lambda-1]$ for
the $i$-th walk.)

  \STATE $v$ calls {\sc Sample-Coupon($v$)} and let $C$ be the
  returned value

  \ENDIF

   \STATE Let $v'$ be node holding coupon $C$. (ID of $v'$ is
   written on $C$.)

  \STATE $v$ sends the token to $v'$, and $v'$ deletes $C$ so that $C$
  will not be sampled again.

  \STATE ${\cal C} = {\cal C} \cup \{v'\}$

\ENDWHILE

\STATE Walk naively until $\ell$ steps are completed (this is at
most another $2\lambda$ steps)

\STATE A node holding the token outputs the ID of $s$

\end{algorithmic}

\end{algorithm}

\paragraph{Faster algorithm}
We are now ready to introduce the second idea which will complete the algorithm. (The complete algorithm is described in Algorithm~\ref{alg:single-random-walk}.) To speed up the above slower algorithm, we pick the length of each short walk uniformly at random in range $[\lambda, 2\lambda-1]$, instead of fixing it to $\lambda$. The reason behind this is that we want every node in the walk to have some probability to take part in token forwarding in Phase~2.

For example, consider running our random walk algorithm on a star network starting at the center and let $\lambda=2$. If all short walks have length two then the center will always forward the token to itself in Phase~2. In other words, the center is the only connector and thus will appear as a connector $\ell/2$ times. This is undesirable since we have to prepare many walks from the center. In contrast, if we randomize the length of each short walk between two and three then the number of times that the center is a connector is $\ell/4$ in expectation. (To see this, observe that, regardless of where the token started, the token will be forwarded to the center with probability $1/2$.)

In the next section, we will show an important property which says that a random walk of length $\ell=O(m^2)$ will
visit each node $v$  at most $\tilde O(\sqrt{\ell}\deg(v))$ times.
We then use the above modification to claim that each node will be
visited as a connector only $\tilde
O(\sqrt{\ell}\deg(v)/\lambda)$ times. This implies that each node
does not have to prepare too many short walks which leads to the
improved running time.

To do this modification, we need to modify Phase~1 and {\sc
Send-More-Coupons}. For Phase~1, we simply change the length of each
short walk to $\lambda+r$ where $r$ is a random integer in $[0,
\lambda-1]$. This modification is shown in
Algorithm~\ref{alg:single-random-walk}. A very slight change is also
made on Phase~2.
For a technical reason, we also prepare $\eta\deg(v)$ coupons from
each node in Phase~1, instead of previously $\eta$ coupons. Our
analysis in the next section shows that this modification still
needs $\tilde O(\eta\lambda)$ rounds as before.

To modify {\sc Send-More-Coupons}, we add Part~2 to the algorithm
(as in Algorithm~\ref{alg:Send-More-Coupons}) where we keep forwarding
each coupon with some probability. It can be shown by a simple
calculation that the number of steps each coupon is forwarded is
uniformly between $\lambda$ and $2\lambda-1$.

We now have the complete description of the algorithm (Algorithm~\ref{alg:single-random-walk}) and are ready
to show the analysis.

\section{Analysis of {\sc Single-Random-Walk}}\label{sec:rw_analysis}

We divide the analysis into four parts. First, we show the
correctness of Algorithm {\sc Single-Random-Walk}. (The proofs of the following
lemmas
will be shown in subsequent sections.)

\newcounter{thmtmp}
\newcounter{counter:correctness}
\setcounter{counter:correctness}{\value{theorem}}
\begin{lemma}\label{lem:correctness}
Algorithm {\sc Single-Random-Walk} solves $1$-RW-DoS. That is, for
any node $v$, after algorithm {\sc Single-Random-Walk} finishes, the
probability that $v$  outputs the ID of $s$ is equal to the
probability that it is the destination of a random walk of length
$\ell$ starting at $s$.
\end{lemma}

Once we have established the correctness, we focus on the running
time. In the second part, we show the probabilistic bound of
Phase~1.

\newcounter{counter:phase1}
\setcounter{counter:phase1}{\value{theorem}}
\begin{lemma} \label{lem:phase1}
Phase~1 finishes in $\tilde O(\lambda \eta)$ rounds with high
probability.
\end{lemma}

In the third part, we analyze the worst case bound of Phase~2, which
is a building block of the probabilistic bound of Phase~2.

\newcounter{counter:phase2-worst}
\setcounter{counter:phase2-worst}{\value{theorem}}
\begin{lemma}\label{lem:phase2-worst}
Phase~2 finishes in $\tilde O(\frac{\ell\cdot
D}{\lambda}+\frac{\ell}{\eta})$ rounds.
\end{lemma}

We note that the above bound holds even when we fix the length of
the short walks (instead of randomly picking from $[\lambda,
2\lambda]$). Moreover, using the above lemmas we can conclude the (weaker)
running time of $\tilde O(\ell^{2/3}D^{1/3})$ by setting $\eta$ and
$\lambda$ appropriately, as follows.

\begin{corollary}
For any $\ell$, Algorithm {\sf Single-Random-Walk} (cf.
Algorithm~\ref{alg:single-random-walk}) solves $1$-RW-DoS correctly
and, with high probability, finishes in $\tilde O(\ell^{2/3}D^{1/3})$
rounds.
\end{corollary}
\begin{proof}
Set $\eta=\ell^{1/3}/D^{1/3}$ and $\lambda=\ell^{1/3}D^{2/3}$. Using
Lemma~\ref{lem:phase1} and \ref{lem:phase2-worst}, the
algorithm finishes in $\tilde O(\lambda \eta+\frac{\ell
D}{\lambda}+\frac{\ell}{\eta})=\tilde O(\ell^{2/3}D^{1/3})$ with high probability.
\end{proof}

In the last part, we improve the running time of Phase~2 further, using a probabilistic bound, leading to a better running time overall. The
key ingredient here is the {\em Random Walk Visits Lemma} (cf.
Lemma~\ref{lemma:visits bound}) stated formally in
Section~\ref{sec:probabilistic_bound} and proved in
Section~\ref{sec:proof-visits-bound}.
Then we use the fact that the short walks have random length to
obtain the running time bound.

\newcounter{counter:phase2-probabilistic}
\setcounter{counter:phase2-probabilistic}{\value{theorem}}
\begin{lemma}\label{lem:phase2-probabilistic}
For any $\eta$ and $\lambda$ such that $\eta\lambda\geq 32\sqrt{\ell}(\log n)^3$, Phase~2 finishes in $\tilde O(\frac{\ell
D}{\lambda})$ rounds with high probability.
\end{lemma}

Using the results above, we conclude the following theorem.

\begin{theorem}\label{thm:single-random-walk}
For any $\ell$, Algorithm {\sf Single-Random-Walk} (cf.
Algorithm~\ref{alg:single-random-walk}) solves $1$-RW-DoS correctly
and, with high probability, finishes in $\tilde O(\sqrt{\ell D})$
rounds.
\end{theorem}
\begin{proof}
Set $\eta=1$ and $\lambda=32\sqrt{\ell D}(\log n)^3$. Using
Lemma~\ref{lem:phase1} and \ref{lem:phase2-probabilistic}, the
algorithm finishes in $\tilde O(\lambda \eta+\frac{\ell
D}{\lambda})=\tilde O(\sqrt{\ell D})$ with high probability.
\end{proof}

\subsection{Correctness (Proof of Lemma~\ref{lem:correctness})}

In this section, we prove Lemma~\ref{lem:correctness} which claims the correctness of the algorithm. Recall that the lemma is as follows.

\setcounter{thmtmp}{\value{theorem}}
\setcounter{theorem}{\value{counter:correctness}}
\begin{lemma}[Restated]
Algorithm {\sc Single-Random-Walk} solves $1$-RW-DoS. That is, for
any node $v$, after algorithm {\sc Single-Random-Walk} finishes, the
probability that $v$  outputs the ID of $s$ is equal to the
probability that it is the destination of a random walk of length
$\ell$ starting at $s$.
\end{lemma}

\setcounter{theorem}{\value{thmtmp}}

To prove this lemma, we first claim that {\sc
Sample-Coupon} returns a coupon where the node holding this coupon
is a destination of a short walk of length uniformly random in
$[\lambda, 2\lambda-1]$.

\begin{claim}
Each short walk length (returned by {\sc Sample-Coupon}) is
uniformly sampled from the range $[\lambda,2\lambda-1]$.
\end{claim}
\begin{proof}
Each walk can be created in two ways.
\begin{itemize}
\item It is created in Phase~1. In this case, since we pick the
length of each walk uniformly from the length
$[\lambda,2\lambda-1]$, the claim clearly holds.
\item It is created by {\sc Send-More-Coupon}. In this case, the claim holds by the
technique of {\em reservoir sampling}~\cite{Vitter85}: Observe that after the
$\lambda^{th}$ step of the walk is completed, we stop extending each
walk at any length between $\lambda$ and $2\lambda-1$ uniformly. To
see this, observe that we stop at length $\lambda$ with probability
$1/\lambda$. If the walk does not stop, it will stop at length
$\lambda+1$ with probability $\frac{1}{\lambda-1}$. This means that
the walk will stop at length $\lambda+1$ with probability
$\frac{\lambda-1}{\lambda}\times \frac{1}{\lambda-1} =
\frac{1}{\lambda}$. Similarly, it can be argue that the walk will
stop at length $i$ for any $i\in [\lambda, 2\lambda-1]$ with
probability $\frac{1}{\lambda}$.
\end{itemize}
\end{proof}

Moreover, we claim that {\sc Sample-Coupon}($v$) samples a short
walk uniformly at random among many coupons (and therefore, short
walks starting at $v$).

\begin{claim}\label{claim:correctness-sample-coupon}
Algorithm {\sc Sample-Coupon}($v$) (cf.
Algorithm~\ref{alg:Sample-Coupon}), for any node $v$, samples a
coupon distributed by $v$ uniformly at random.
\end{claim}
\begin{proof}
Assume that before this algorithm starts, there are  $t$ (without
loss of generality, let $t > 0$) coupons containing ID of $v$ stored
in some nodes in the network. The goal is to show that {\sc
Sample-Coupon} brings one of these coupons to $v$ with uniform
probability. For any node $u$, let $T_u$ be the subtree rooted at
$u$ and let $S_u$ be the set of coupons in $T_u$. (Therefore,
$T_v=T$ and $|S_v|=t$.)

We claim that any node $u$ returns a coupon to its parent with
uniform probability (i.e., for any coupons $x\in S_u$, $\mathbb{P}[ u$
returns $x ]$ is $1/|S_u|$ (if $|S_u|>0$)). We prove this by
induction on the height of the tree. This claim clearly holds for
the base case where $u$ is a leaf node. Now, for any non-leaf node
$u$, assume that the claim is true for any of its children.
To be precise, suppose that $u$ receives coupons and counts from $q-1$
children. Assume that it receives coupons $d_1, d_2, ..., d_{q-1}$ and
counts $c_1, c_2, ..., c_{q-1}$ from nodes $u_1, u_2, ..., u_{q-1}$,
respectively. (Also recall that $d_0$ is the sample of its own
coupons (if exists) and $c_0$ is the number of its own coupons.) By
induction, $d_j$ is sent from $u_j$ to $u$ with probability
$1/|S_{u_j}|$, for any $0\leq j\leq q-1$. Moreover, $c_j=|S_{u_j}|$
for any $j$. Therefore, any coupon $d_j$ will be picked with
probability $\frac{1}{|S_{u_j}|}\times \frac{c_j}{c_0+c_1+...c_{q-1}} =
\frac{1}{|S_u|}$ as claimed.

The lemma follows by applying the claim above to $v$.
\end{proof}

The above two claims imply the correctness of the Algorithm {\sf Single-Random-Walk} as shown next.
\begin{proof}[of Lemma \ref{lem:correctness}]
Any two $[\lambda,2\lambda-1]$-length walks (possibly
from different sources) are independent from each other. Moreover, a
walk from a particular node is picked uniformly at random.
Therefore, algorithm {\sf Single-Random-Walk} is equivalent to
having a source node perform a walk of length between $\lambda$ and
$2\lambda-1$ and then have the destination do another walk of length
between $\lambda$ and $2\lambda-1$ and so on. That is, for
any node $v$, the
probability that $v$  outputs the ID of $s$ is equal to the
probability that it is the destination of a random walk of length
$\ell$ starting at $s$.
\end{proof}

\subsection{Analysis of Phase~1 (Proof of Lemma \ref{lem:phase1})}

In this section, we prove the performance of Phase~1 claimed in Lemma~\ref{lem:phase1}. Recall that the lemma is as follows.

\setcounter{thmtmp}{\value{theorem}}
\setcounter{theorem}{\value{counter:phase1}}
\begin{lemma}[Restated]
Phase~1 finishes in $\tilde O(\lambda \eta)$ rounds with high
probability.
\end{lemma}

\setcounter{theorem}{\value{thmtmp}}

We now prove the lemma. For each coupon $C$, any $j=1, 2, ..., \lambda$, and any edge $e$,
we define $X_C^j(e)$ to be a random variable having value 1 if $C$
is sent through $e$ in the $j^{th}$ iteration (i.e., when the
counter on $C$ is increased from $j-1$ to $j$). Let $X^j(e)=\sum_{C:
\text{coupon}} X_C^j(e)$.  We compute the expected number of coupons
that go through an edge $e$, as follows.

\begin{claim}
\label{claim:first} For any edge $e$ and any $j$,
$\mathbb{E}[X^j(e)]=2\eta$.
\end{claim}
\begin{proof}
Recall that each node $v$ starts with $\eta \deg(v)$ coupons and
each coupon takes a random walk. We prove that after any given
number of steps $j$, the expected number of coupons at node $v$ is
still $\eta \deg(v)$. Consider the random walk's probability
transition matrix, call it $A$. In this case $Au = u$ for the vector
$u$ having value $\frac{\deg(v)}{2m}$ where $m$ is the number of
edges in the graph (since this $u$ is the stationary distribution of
an undirected unweighted graph). Now the number of coupons we
started with at any node $i$ is proportional to its stationary
distribution, therefore, in expectation, the number of coupons at
any node remains the same.

To calculate $\mathbb{E}[X^j(e)]$, notice that edge $e$ will receive
coupons from its two end points, say $x$ and $y$. The number of
coupons it receives from node $x$ in expectation is exactly the
number of coupons at $x$ divided by $\deg(x)$. The claim follows.
\end{proof}

By Chernoff's bound (e.g., in~\cite[Theorem~4.4.]{MU-book-05}), for
any edge $e$ and any $j$,
$$\mathbb{P}[X^j(e)\geq 4\eta\log{n}]\leq 2^{-4\log{n}}=n^{-4}.$$
(We note that the number $4\eta\log {n}$ above can be improved to $c\eta\log{n}/\log\log n$ for some constant $k$. This improvement of $\log\log n$ can be further improved as $\eta$ increases. This fact is useful in practice but does not help improve our claimed running time since we always hide a $\polylog{n}$ factor.)

It follows that the probability that there exists an edge $e$ and an
integer $1\leq j\leq \lambda$ such that $X^j(e)\geq 4\eta\log{n}$ is
at most $|E(G)| \lambda n^{-4}\leq \frac{1}{n}$ since $|E(G)|\leq
n^2$ and $\lambda\leq \ell\leq n$ (by the way we define $\lambda$).

Now suppose that $X^j(e)\leq 4\eta\log{n}$ for every edge $e$ and
every integer $j\leq \lambda$. This implies that we can extend all
walks of length $i$ to length $i+1$ in $4\eta\log{n}$ rounds.
Therefore, we obtain walks of length $\lambda$ in
$4\lambda\eta\log{n}$ rounds, with high probability, as claimed.

\subsection{Worst-case bound of Phase~2 (Proof of Lemma~\ref{lem:phase2-worst})}

In this section, we prove the {\em worst-case} performance of Phase~2 claimed in Lemma~\ref{lem:phase2-worst}. Recall that the lemma is as follows.

\setcounter{thmtmp}{\value{theorem}}
\setcounter{theorem}{\value{counter:phase2-worst}}
\begin{lemma}[Restated]
Phase~2 finishes in $\tilde O(\frac{\ell\cdot
D}{\lambda}+\frac{\ell}{\eta})$ rounds.
\end{lemma}

\setcounter{theorem}{\value{thmtmp}}

We first analyze the running
time of {\sc Send-More-Coupons} and {\sc Sample-Coupon}.

\begin{lemma}\label{lem:Send-More-coupons}
For any $v$, {\sc Send-More-Coupons}($v$, $\eta$, $\lambda$) always
finishes within $O(\lambda)$ rounds.
\end{lemma}

\begin{proof}
Consider any node $u$ during the execution of the algorithm. If it
contains $x$ coupons of $v$ (i.e., which just contain the ID of
$v$), for some $x$, it has to pick $x$ of its neighbors at random,
and pass the coupon of $v$ to each of these $x$ neighbors. It might
pass these coupons to less than $x$ neighbors and cause congestion
if the coupons are sent separately. However, it sends only the ID of
$v$ and a {\em count} to each neighbor, where the count represents
the number of coupons it wishes to send to such neighbor. Note that
there is only one ID sent during the process since only one node
calls {\sc Send-More-Coupons} at a time. Therefore, there is no
congestion and thus the algorithm terminates in $O(\lambda)$ rounds.
\end{proof}

\begin{lemma}\label{lem:Sample-Coupon}
{\sc Sample-Coupon} always finishes within $O(D)$ rounds.
\end{lemma}

\begin{proof}
Since, constructing a BFS tree can be done easily in $O(D)$ rounds,
it is left to bound the time of the second part where the algorithm
wishes to {\em sample} one of many coupons (having its ID) spread
across the graph. The sampling is done while retracing the BFS tree
starting from leaf nodes, eventually reaching the root. The main
observation is that when a node receives multiple samples from its
children, it only sends one of them to its parent. Therefore, there
is no congestion. The total number of rounds required is therefore
the number of levels in the BFS tree, $O(D)$.
\end{proof}

Now we prove the worst-case bound of Phase~2. First, observe that
{\sc Sample-Coupon} is called $O(\frac{\ell}{\lambda})$ times since
it is called only by a connector (to find the next node to forward
the token to). By Lemma~\ref{lem:Sample-Coupon}, this algorithm
takes $O(\frac{\ell\cdot D}{\lambda})$ rounds in total.
Next, we claim that {\sc Send-More-Coupons} is called at most
$O(\frac{\ell}{\lambda\eta})$ times in total (summing over all
nodes). This is because when a node $v$ calls {\sc
Send-More-Coupons}($v$, $\eta$, $\lambda$), all $\eta$ walks
starting at $v$ must have been stitched and therefore $v$
contributes $\lambda\eta$ steps of walk to the long walk we are
constructing.
It follows from Lemma~\ref{lem:Send-More-coupons} that {\sc
Send-More-Coupons} algorithm takes $O(\frac{\ell}{\eta})$ rounds in
total. The claimed worst-case bound follows by summing up the total
running times of {\sc Sample-Coupon} and {\sc Send-More-Coupons}.

\subsection{A Probabilistic bound for Phase~2 (Proof of Lemma~\ref{lem:phase2-probabilistic})}\label{sec:probabilistic_bound}

In this section, we prove the {\em high probability} time bound of Phase~2 claimed in Lemma~\ref{lem:phase2-probabilistic}. Recall that the lemma is as follows.

\setcounter{thmtmp}{\value{theorem}}
\setcounter{theorem}{\value{counter:phase2-probabilistic}}
\begin{lemma}[Restated]
For any $\eta$ and $\lambda$ such that $\eta\lambda\geq 32\sqrt{\ell}(\log n)^3$, Phase~2 finishes in $\tilde O(\frac{\ell
D}{\lambda})$ rounds with high probability.
\end{lemma}
\setcounter{theorem}{\value{thmtmp}}

Recall that we may assume that $\ell=O(m^2)$ (cf. Observation \ref{obs:length at
most m2}).
We prove the stronger bound using the following lemmas. As mentioned
earlier, to bound the number of times {\sc Send-More-Coupons} is
invoked, we need a technical result on random walks that bounds the
number of times a node will be visited in a $\ell$-length random
walk. Consider a simple random walk on a connected undirected graph
on $n$ vertices. Let $\deg(x)$ denote the degree of $x$, and let $m$
denote the number of edges. Let $N_t^x(y)$ denote the number of
visits to vertex $y$ by time $t$, given that the walk started at vertex
$x$.

Now, consider $k$ walks, each of length $\ell$, starting from (not
necessary distinct) nodes $x_1, x_2, \ldots, x_k$. We show a key
technical lemma that applies to random walks on any (undirected) graph: With
high probability, no vertex $y$ is visited more than $32 \deg(x)
\sqrt{k\ell+1}\log n + k$ times.

\newcounter{counter:rw_visits}
\setcounter{counter:rw_visits}{\value{theorem}}
\begin{lemma}[Random Walk Visits Lemma]\label{lemma:visits bound}
For any nodes $x_1, x_2, \ldots, x_k$, and $\ell=O(m^2)$, \[\Pr\bigl(\exists y\ s.t.\
\sum_{i=1}^k N_\ell^{x_i}(y) \geq 32 \deg(x) \sqrt{k\ell+1}\log
n+k\bigr) \leq 1/n\,.\]
\end{lemma}

Since the proof of this lemma is interesting on its own and
lengthy, we defer it to Section~\ref{sec:proof-visits-bound}.
We note that one can also show a similar bound for a specific vertex, i.e. $\Pr\bigl(\exists y\ s.t.\
\sum_{i=1}^k N_\ell^{x_i}(y) \geq 32 \deg(x) \sqrt{k\ell+1}\log
n+k\bigr)$. Since we will not use this bound here, we defer it to Lemma~\ref{lemma:k walks one node bound} in Subsection~\ref{sec:proof-visits-bound}.
%
%
Moreover, we prove the above lemma only for a specific number of visits of roughly $\sqrt{k\ell}$ because this is the expected number of visits (we show this in Proposition~\ref{proposition:first and second moment} in Section~\ref{sec:proof-visits-bound}). It might be possible to prove more general bounds; however, we do not include them here since they need more proofs and are not relevant to the results of this paper.


Also note that Lemma~\ref{lemma:visits bound} is not true if we do not restrict $\ell$ to be $O(m^2)$. For example, consider a star network and a walk of length $\ell$ such that $\ell\gg n^2$ and $\ell$ is larger than the mixing time. In this case, this walk will visit the center of the star $\tilde \Omega(\ell)$ times with high probability. This contradicts Lemma~\ref{lemma:visits bound} which says that the center will be visited $\tilde O(n\sqrt{\ell})=o(\ell)$ times with high probability.
We can modify the statement of Lemma~\ref{lemma:visits bound} to hold for a general value of $\ell$ as follows (this fact is not needed in this paper):
$\Pr(\exists y\ s.t.\
\sum_{i=1}^k N_\ell^{x_i}(y) \geq 32 \deg(x) \sqrt{k\ell+1}\log
n+k + \ell\deg(x)/m) \leq 1/n.$
(Recall that $m$ is the number of edges in the network.) This inequality can be proved using Lemma~\ref{lemma:visits bound} and the fact that $m^2$ is larger than the mixing time, which means that the walk will visit vertex $x$ with probability $\deg(x)/m$ in each step after the $(m^2)^{th}$ step.

Lemma~\ref{lemma:visits bound} says that the number of visits to each node can be
bounded. However, for each node, we are only interested in the case
where it is used as a connector. The lemma below shows that the
number of visits as a connector can be bounded as well; i.e.,
if any node $v_i$ appears $t$ times in the walk, then it is likely
to appear roughly $t/\lambda$ times as connectors.

\begin{lemma}
\label{lem:uniformityused} For any vertex $v$, if $v$ appears in the
walk at most $t$ times then it appears as a connector node at most
$t(\log n)^2/\lambda$ times with probability at least $1-1/n^2$.
\end{lemma}

At first thought, the lemma above might sound correct even when we
do not randomize the length of the short walks since the connectors are
spread out in steps of length approximately $\lambda$. However,
there might be some {\em periodicity} that results in the same node
being visited multiple times but {\em exactly} at
$\lambda$-intervals.
(As we described earlier, one example is when the input network is a star graph and $\lambda=2$.)
This is where we crucially use the fact that the algorithm uses
walks of length uniformly random in $[\lambda, 2\lambda-1]$. The
proof then goes via constructing another process equivalent to
partitioning the $\ell$ steps into intervals of $\lambda$ and then
sampling points from each interval. We analyze this by
constructing a different process that stochastically dominates the
process of a node occurring as a connector at various steps in the
$\ell$-length walk and then use a Chernoff bound argument.

In order to give a detailed proof of Lemma~\ref{lem:uniformityused},
we need the following two claims.


\begin{claim}\label{claim:BC}
Consider any sequence $A$ of numbers $a_1, ..., a_{\ell'}$
of length $\ell'$. For any integer $\lambda'$, let $B$ be a sequence
$a_{\lambda'+r_1}, a_{2\lambda'+r_1+r_2}, ...,
a_{i\lambda'+r_1+...+r_i}, ...$ where $r_i$, for any $i$, is a
random integer picked uniformly from $[0, \lambda'-1]$. Consider
another subsequence of numbers $C$ of $A$ where an element in $C$ is
picked from ``every $\lambda'$ numbers'' in $A$; i.e., $C$ consists
of $\lfloor\ell'/\lambda'\rfloor$ numbers $c_1, c_2, ...$ where, for
any $i$, $c_i$ is chosen uniformly at random from
$a_{(i-1)\lambda'+1}, a_{(i-1)\lambda'+2}, ..., a_{i\lambda'}$.
Then, $\mathbb{P}[C \text{ contains } \{a_{i_1}, a_{i_2}, ..., a_{i_k}\}] =
\mathbb{P}[B = \{a_{i_1}, a_{i_2}, ..., a_{i_k}\}]$ for any set $\{a_{i_1},
a_{i_2}, ..., a_{i_k}\}$.
\end{claim}

\begin{proof}
Observe that $B$ will be equal to $\{a_{i_1}, a_{i_2},
..., a_{i_k}\}$ only for a specific value of $r_1, r_2, ..., r_k$.
Since each of $r_1, r_2, ..., r_k$ is chosen uniformly at random
from $[1, \lambda']$, $\mathbb{P}[B = \{a_{i_1}, a_{i_2}, ..., a_{i_k}\}] =
\lambda'^{-k}$.
Moreover, the $C$ will contain $a_{i_1}, a{i_2}, ..., a_{i_k}\}$ if
and only if, for each $j$, we pick $a_{i_j}$ from the interval that
contains it (i.e., from $a_{(i'-1)\lambda'+1}, a_{(i'-1)\lambda'+2},
..., a_{i'\lambda'}$, for some $i'$). (Note that $a_{i_1}, a_{i_2},
...$ are all in different intervals because $i_{j+1}-i_j\geq
\lambda'$ for all $j$.) Therefore, $\mathbb{P}[C \text{ contains } a_{i_1},
a_{i_2}, ..., a_{i_k}\}]=\lambda'^{-k}$.
\end{proof}

\begin{claim}\label{claim:C_bound}
Consider any sequence $A$ of numbers $a_1, ..., a_\ell'$ of length
$\ell'$. Consider subsequence of numbers $C$ of $A$ where an element
in $C$ is picked from from ``every $\lambda'$ numbers'' in $A$;
i.e., $C$ consists of $\lfloor\ell'/\lambda'\rfloor$ numbers $c_1,
c_2, ...$ where, for any $i$, $c_i$ is chosen uniformly at random
from $a_{(i-1)\lambda'+1}, a_{(i-1)\lambda'+2}, ...,
a_{i\lambda'}$.. For any number $x$, let $n_x$ be the number of
appearances of $x$ in $A$; i.e., $n_x=|\{i\ |\ a_i=x\}|$. Then, for
any $R\geq 6n_x/\lambda'$, $x$ appears in $C$ more than $R$ times
with probability at most $2^{-R}$.
\end{claim}
\begin{proof}
For $i=1, 2, ..., \lfloor\ell'/\lambda'\rfloor$, let $X_i$ be a 0/1
random variable that is $1$ if and only if $c_i=x$ and
$X=\sum_{i=1}^{\lfloor\ell'/\lambda'\rfloor} X_i$. That is, $X$ is
the number of appearances of $x$ in $C$. Clearly,
$E[X]=n_x/\lambda'$. Since $X_i$'s are independent, we can apply the
Chernoff bound (e.g., in~\cite[Theorem~4.4.]{MU-book-05}): For any
$R\geq 6E[X]=6n_x/\lambda'$,
$$\mathbb{P}[X\leq R]\geq 2^{-R}.$$
The claim is thus proved.
\end{proof} 

\begin{proof}[of Lemma~\ref{lem:uniformityused}] Now we use the claims to prove the lemma. Choose $\ell'=\ell$
and $\lambda'=\lambda$ and consider any node $v$ that appears at
most $t$ times. The number of times it appears as a connector node
is the number of times it appears in the subsequence $B$ described
in Claim~\ref{claim:BC}. By applying Claim~\ref{claim:BC} and \ref{claim:C_bound} with $R=t(\log n)^2$, we have
that $v$ appears in $B$ more than $t(\log n)^2$ times with
probability at most $1/n^2$ as desired.
\end{proof}

Now we are ready to prove the probabilistic bound of Phase~2 (cf.
Lemma~\ref{lem:phase2-probabilistic}).

First, we claim, using Lemma \ref{lemma:visits bound} and
\ref{lem:uniformityused}, that each node is used as a connector node
at most $\frac{32 \deg(x) \sqrt{\ell}(\log n)^3}{\lambda}$ times with
probability at least $1-2/n$. To see this, observe that the claim
holds if each node $x$ is visited at most
$t(x)=32\deg(x)\sqrt{\ell+1}\log n$ times and consequently appears as a
connector node at most $t(x)(\log n)^2/\lambda$ times. By
Lemma~\ref{lemma:visits bound}, the first condition holds with
probability at least $1-1/n$. By Lemma~\ref{lem:uniformityused} and
the union bound over all nodes, the second condition holds with
probability at least $1-1/n$, provided that the first condition
holds. Therefore, both conditions hold together with probability at
least $1-2/n$ as claimed.

Now, observe that {\sc Sample-Coupon} is invoked
$O(\frac{\ell}{\lambda})$ times (only when we stitch the walks) and
therefore, by Lemma~\ref{lem:Sample-Coupon}, contributes
$O(\frac{\ell D}{\lambda})$ rounds.
Moreover, we claim that {\sc Send-More-Coupons} is never invoked,
with probability at least $1-2/n$. To see this, recall our claim
above that each node $x$ is used as a connector node at most $\frac{32
\deg(x) \sqrt{\ell}(\log n)^3}{\lambda}$ times. Additionally, observe
that we have prepared this many walks in Phase~1; i.e., after
Phase~1, each node has $\eta \deg(x)\geq \frac{32 \deg(x)
\sqrt{\ell}(\log n)^3}{\lambda}$ short walks. The claim follows.

Therefore, with probability at least $1-2/n$, the rounds are $\tilde
O(\frac{\ell D}{\lambda})$ as claimed.\danupon{I found many errors in Section~\ref{sec:probabilistic_bound}. It should be read carefully once more.}

\subsection{Proof of Random Walk Visits Lemma (cf. Lemma~\ref{lemma:visits
bound})}\label{sec:proof-visits-bound}
In this section, we prove the Random Walk Visits Lemma introduced in the previous section. We restated it here for the sake of readability.

\setcounter{thmtmp}{\value{theorem}}
\setcounter{theorem}{\value{counter:rw_visits}}
\begin{lemma}[Random Walk Visits Lemma, Restated]\label{lemma:visits bound}
For any nodes $x_1, x_2, \ldots, x_k$, and $\ell=O(m^2)$, \[\Pr\bigl(\exists y\ s.t.\
\sum_{i=1}^k N_\ell^{x_i}(y) \geq 32 \deg(x) \sqrt{k\ell+1}\log
n+k\bigr) \leq 1/n\,.\]
\end{lemma}
\setcounter{theorem}{\value{thmtmp}}

We start with the bound of the first moment of the number of
visits at each node by each walk.

\begin{proposition}\label{proposition:first and second moment} For
any node $x$, node $y$ and $t = O(m^2)$,
\begin{equation}
\e[N_t^x(y)] \le 8 \deg(y) \sqrt{t+1}
\,.
\end{equation}
\end{proposition}

To prove the above proposition, let $P$ denote the transition
probability matrix of such a random walk and let $\pi$ denote the
stationary distribution of the walk, which in this case is simply
proportional to the degree of the vertex, and let $\pi_\m = \min_x
\pi(x)$.

The basic bound we use is the following estimate from Lyons (see
Lemma~3.4 and Remark~4  in \cite{Lyons}). Let $Q$ denote the
transition probability matrix of a chain with self-loop probablity
$\alpha > 0$, and with $c= \min{\{\pi(x) Q(x,y) : x\neq y \mbox{ and
}                      Q(x,y)>0\}}\,.$ Note that for a random walk
on an undirected graph, $c=\frac{1}{2m}$. For $k > 0$ a positive
integer (denoting time) ,

\begin{equation}
\label{kernel_decay} \bigl|\frac{Q^k(x,y)}{\pi(y)} - 1\bigr| \le
\min\Bigl\{\frac{1}{\alpha c \sqrt{k+1}}, \frac{1}{2\alpha^2 c^2
(k+1)} \Bigr\}\,.
\end{equation}

For $k\leq\beta m^2$ for a sufficiently small constant $\beta$, and
small $\alpha$, the above can be simplified to the following bound (we use Observation~\ref{obs:length at most m2} here);
see Remark~3 in \cite{Lyons}.
\begin{equation}
\label{one_sided_decay} Q^k(x,y)  \le \frac{4\pi(y)}{c \sqrt{k+1}} =
\frac{4\deg(y)}{\sqrt{k+1}}\,.
\end{equation}

Note that given a simple random walk on a graph $G$, and a
corresponding matrix  $P$, one can always switch to the lazy version
$Q=(I+P)/2$, and interpret it as a walk on graph $G'$, obtained by
adding  self-loops  to vertices in $G$ so as to double the degree of
each vertex. In the following, with abuse of notation we assume our
$P$ is such a lazy version of the original one.

\begin{proof}[of Proposition~\ref{proposition:first and second moment}]
Let $X_0, X_1, \ldots $ describe the random walk, with $X_i$
denoting the position of the walk at time $i\ge 0$, and let
$\bone_A$ denote the indicator (0-1) random variable, which takes
the value 1 when the event $A$ is true. In the following we also use
the subscript $x$ to denote the fact that the probability or
expectation is with respect to starting the walk at vertex $x$.
We get the expectation.
\begin{eqnarray*}
\e[N_t^x(y)] & = & \e_x[  \sum_{i=0}^t \bone_{\{X_i=y\}}] = \sum_{i=0}^t P^i(x,y) \\
& \le &  4 \deg(y) \sum_{i=0}^t \frac{1}{\sqrt{i+1}} , \ \ \mbox{ (using the above inequality  (\ref{one_sided_decay})) } \\
& \le & 8 \deg(y) \sqrt{t+1}\,.
\end{eqnarray*}
%
%
%
\end{proof}

Using the above proposition, we bound the number of visits of each
walk at each node, as follows.

\begin{lemma}\label{lemma:whp one walk one node bound}
For $t=O(m^2)$ and any vertex $y \in G$, the random walk started at
$x$ satisfies:
\begin{equation*}
\Pr\bigl(N^x_t(y) \ge  32  \ \deg(y) \sqrt{t+1}\log n \bigr) \le
\frac{1}{n^2} \,.
\end{equation*}
\end{lemma}
\begin{proof}
First, it follows from the Proposition and Markov's inequality that
\begin{equation} \Pr\bigl(N^x_t(y) \ge  4\cdot 8 \
\deg(y) \sqrt{t+1}\bigr) \le \frac{1}{4} \,.\label{eq:simple bound}
\end{equation}
%

%

%
For any $r$, let $L^x_r(y)$ be the time that the random walk
(started at $x$) visits $y$ for the $r^{th}$ time. Observe that, for
any $r$, $N^x_t(y)\geq r$ if and only if $L^x_r(y)\leq t$.
Therefore,
\begin{equation}
\Pr(N^x_t(y)\geq r)=\Pr(L^x_r(y)\leq t).\label{eq:visits eq length}
\end{equation}

Let $r^*=32  \ \deg(y) \sqrt{t+1}$. By \eqref{eq:simple bound} and
\eqref{eq:visits eq length}, $\Pr(L^x_{r^*}(y)\leq t)\leq
\frac{1}{4}\,.$ We claim that
\begin{equation}
\Pr(L^x_{r^*\log n}(y)\leq t)\leq \left(\frac{1}{4}\right)^{\log
n}=\frac{1}{n^2}\,.\label{eq:hp length bound}
\end{equation}
To see this, divide the walk into $\log n$ independent subwalks,
each visiting $y$ exactly $r^*$ times. Since the event $L^x_{r^*\log
n}(y)\leq t$ implies that all subwalks have length at most $t$,
\eqref{eq:hp length bound} follows.
Now, by applying \eqref{eq:visits eq length} again,
\[\Pr(N^x_t(y)\geq r^*\log n) = \Pr(L^x_{r^*\log n}(y)\leq t)\leq
\frac{1}{n^2}\]
as desired.
\end{proof}

We now extend the above lemma to bound the number of visits of {\em
all} the walks at each particular node.
\begin{lemma}[Random Walk Visits Lemma For a Specific Vertex]\label{lemma:k walks one node bound}
For $\gamma > 0$, and $t=O(m^2)$, and for any vertex $y \in G$, the
random walk started at $x$ satisfies:
\begin{equation*}
\Pr\bigl(\sum_{i=1}^k N^{x_i}_t(y) \ge  32  \ \deg(y) \sqrt{kt+1} \log
n+k\bigr) \le \frac{1}{n^2} \,.
\end{equation*}
\end{lemma}
\begin{proof}
First, observe that, for any $r$,
\begin{equation}
\Pr\bigl(\sum_{i=1}^k
N^{x_i}_t(y) \geq r-k\bigr)\leq \mathbb{P}[N^y_{kt}(y)\geq r].\label{eq: k walks one node arbitrary visit bound}
\end{equation}
To see
this, we construct a walk $W$ of length $kt$ starting at $y$ in the
following way: For each $i$, denote a walk of length $t$ starting at
$x_i$ by $W_i$. Let $\tau_i$ and $\tau'_i$ be the first and last
time (not later than time $t$) that $W_i$ visits $y$. Let $W'_i$ be
the subwalk of $W_i$ from time $\tau_i$ to $\tau_i'$. We construct a
walk $W$ by stitching $W'_1, W'_2, ..., W'_k$ together and complete
the rest of the walk (to reach the length $kt$) by a normal random
walk. It then follows that the number of visits to $y$ by $W_1, W_2,
\ldots, W_k$ (excluding the starting step) is at most the number of
visits to $y$ by $W$. The first quantity is $\sum_{i=1}^k
N^{x_i}_t(y)-k$. (The term `$-k$' comes from the fact that we do not
count the first visit to $y$ by each $W_i$ which is the starting
step of each $W'_i$.) The second quantity is $N^y_{kt}(y)$. The
observation thus follows.

Therefore, \[\Pr\bigl(\sum_{i=1}^k N^{x_i}_t(y)\geq 32 \ \deg(y)
\sqrt{kt+1}\log n + k\bigr) \leq \Pr\bigl(N^y_{kt}(y)\geq 32 \ \deg(y)
\sqrt{kt+1}\log n\bigr) \leq \frac{1}{n^2}\]
where the last inequality follows from Lemma~\ref{lemma:whp one walk
one node bound}.
%
%
\end{proof}

The Random Walk Visits Lemma (cf. Lemma~\ref{lemma:visits bound})
follows immediately from Lemma~\ref{lemma:k walks one node bound} by
union bounding over all nodes.

\section{Variations, Extensions, and Generalizations}

\subsection{Computing $k$ Random Walks}\label{subsec:many walks}

We now consider the scenario when we want to compute $k$ walks of
length $\ell$ from different (not necessary distinct) sources $s_1,
s_2, \ldots, s_k$. We show that {\sc Single-Random-Walk} can be
extended to solve this problem. Consider the following  algorithm.

\paragraph{{\sc Many-Random-Walks}} Let
$\lambda=(32 \sqrt{k\ell D+1}\log n+k)(\log n)^2$ and $\eta=1$. If
$\lambda> \ell$ then run the naive random walk algorithm, i.e., the
sources find walks of length $\ell$ simultaneously by sending
tokens. Otherwise, do the following. First, modify Phase~2 of {\sc
Single-Random-Walk} to create multiple walks, one at a time; i.e.,
in the second phase, we stitch the short walks together to get a
walk of length $\ell$ starting at $s_1$ then do the same thing for
$s_2$, $s_3$, and so on.

The correctness of {\sc Many-Random-Walks} follows from Lemma~\ref{lem:correctness}; intuitively, this algorithm outputs independent random walks because it obtains long walks by stitching short walks that are all independent (no short walk is used twice). We now prove the running time of this algorithm.

\begin{theorem}\label{thm:kwalks} {\sc Many-Random-Walks} finishes in
$\tilde O\left(\min(\sqrt{k\ell D}+k, k+\ell)\right)$ rounds with
high probability.
\end{theorem}

\begin{proof}
First, consider the case where $\lambda>\ell$. In this case,
$\min(\sqrt{k\ell D}+k, \sqrt{k\ell}+k+\ell)=\tilde
O(\sqrt{k\ell}+k+\ell)$. By Lemma~\ref{lemma:visits bound}, each
node $x$ will be visited at most $\tilde O(\deg(x) (\sqrt{k\ell}+k))$
times. Therefore, using the same argument as Lemma~\ref{lem:phase1},
the congestion is $\tilde O(\sqrt{k\ell} + k)$ with high
probability. Since the dilation is $\ell$, {\sc Many-Random-Walks}
takes $\tilde O(\sqrt{k\ell} + k +\ell)$ rounds as claimed. Since
$2\sqrt{k\ell}\leq k+\ell$, this bound reduces to $O(k+\ell)$.

Now, consider the other case where $\lambda\leq \ell$. In this case,
$\min(\sqrt{k\ell D}+k, \sqrt{k\ell}+k+\ell)=\tilde O(\sqrt{k\ell
D}+k)$. Phase~1 takes $\tilde O(\lambda \eta) = \tilde O(\sqrt{k\ell
D}+k)$. The stitching in Phase~2 takes $\tilde O(k\ell D/\lambda) =
\tilde O(\sqrt{k\ell D})$. Moreover, by Lemma~\ref{lemma:visits
bound}, {\sc send-more-coupons} will never be invoked. Therefore,
the total number of rounds is $\tilde O(\sqrt{k\ell D}+k)$ as
claimed.
\end{proof}

\subsection{Regenerating the entire random walk}

Our
algorithm can be extended to regenerate the entire walk, solving
$k$-RW-pos. This will be use, e.g., in generating a random spanning
tree.
The algorithm is the following. First, inform all intermediate
connecting nodes of their position which can be done by keeping
track of the walk length when we do token forwarding in Phase~2.
Then, these nodes can regenerate their $O(\sqrt{\ell})$ length short
walks by simply sending a message through each of the corresponding
short walks. This can be completed in $\tilde{O}(\sqrt{\ell D})$
rounds with high probability. This is because, with high
probability, {\sc Send-More-Coupons} will not be invoked and hence
all the short walks are generated in Phase~1. Sending a message
through each of these short walks (in fact, sending a message
through {\em every} short walk generated in Phase~1) takes time at
most the time taken in Phase~1, i.e., $\tilde{O}(\sqrt{\ell D})$
rounds.

\subsection{Generalization to the Metropolis-Hastings algorithm}


We now discuss extensions of our algorithm to perform a random walk
according to  the Metropolis-Hastings algorithm, a more general type
of random walk with numerous applications (e.g., \cite{ZS06}). The
Metropolis-Hastings \cite{Hastings70,MRRT53} algorithm gives a way
to define a transition probability so that a random walk converges
to any desired distribution $\pi$ (where $\pi_i$, for any node $i$,
is the desired stationary probability at node $i$). It is assumed
that every node $i$ knows its steady state probability $\pi_i$ (and
can know its neighbors' steady state probabilities in one round).

The Metropolis-Hastings algorithm is roughly as follows (see, e.g., \cite{Hastings70,MRRT53} for the full description). For any desired distribution $\pi$ and any desired \textit{laziness factor} $0<\alpha<1$, the transition probability from node $i$ to its neighbor $j$ is defined to be
$$P_{ij}=\alpha\min(1/d_i, \pi_j/(\pi_i d_j))$$
where $d_i$ and $d_j$ are degree of $i$ and $j$ respectively. It can be shown that a random walk with this transition probability converges to $\pi$.

Using the transition probability defined above, we now run the {\sc Single-Random-Walk} algorithm with one modification: in Phase~1, we generate
$$\eta\cdot \frac{\pi(x)}{\alpha \min_x \frac{\pi(x)}{\deg(x)}}$$
short walks instead of $\eta \deg(v)$.
\gopal{Also, we should be consistent with respect to using $x$ or $v$. So
we should use $v$ in place of $x$.}
\danupon{Need everyone to help on this. There are too many place.}

The correctness of the algorithm follows from Lemma~\ref{lem:correctness}. The running time follows from the following theorem.

%

\begin{theorem}\label{thm:metropolis-hasting}
For any $\eta$ and $\lambda$ such that $\eta\lambda\geq 32\sqrt{\ell}(\log n)^3$, the modified {\sc Single-Random-Walk} algorithm stated above finishes in
$$\tilde O(\lambda\eta\cdot \frac{\max_x \pi(x)/\deg(x)}{\min_y \pi(y)/\deg(y)} + \frac{\ell D}{\lambda})$$
rounds with high probability.
\end{theorem}
%

An interesting application of the above theorem is when $\pi$ is a stationary distribution. In this case, we can compute a random walk of length $\ell$ in $\tilde O(\lambda \eta+\frac{\ell D}{\lambda})$ rounds which is exactly Theorem~\ref{thm:single-random-walk}.
Like Theorem~\ref{thm:single-random-walk}, the above theorem follows from the following two lemmas which are similar to Lemmas~\ref{lem:phase1} and \ref{lem:phase2-probabilistic}.


\begin{lemma} \label{lem:phase1-Metropolis-Hastings}
For any $\pi$ and $\alpha$, Phase~1 finishes in $O(\lambda\eta\log n \cdot \frac{\max_x \pi(x)/\deg(x)}{\min_y \pi(y)/\deg(y)})$
rounds with high probability.
\end{lemma}
\begin{proof}
%
%
The proof is essentially the same as Lemma~\ref{lem:phase1}.
We present it here for completeness.
Let $\beta=\frac{1}{\alpha \min_x \frac{\pi(x)}{\deg(x)}}$.
Consider the case when each node $i$ creates $\beta\pi(i)\eta$ messages.
We show that the lemma  holds even in this case.

We use the same definition as in Lemma~\ref{lem:phase1}.
That is, for each message $M$, any $j=1, 2, ..., \lambda$, and any edge $e$,
we define $X_M^j(e)$ to be a random variable having value 1 if $M$
is sent through $e$ in the $j^{th}$ iteration (i.e., when the counter on $M$ has value $j-1$).
Let $X^j(e)=\sum_{M: \text{message}} X_M^j(e)$.  We compute the expected number of messages that go through an edge.
As before, we show the following claim.

\begin{claim}
For any edge $e$ and any $j$, $\mathbb{E}[X^j(e)]=
2\eta\cdot \frac{\max_x \pi(x)/\deg(x)}{\min_y \pi(y)/\deg(y)}$.
\end{claim}
\begin{proof}
Assume that each node $v$ starts with $\beta\pi(v)\eta$
messages. Each message takes a random walk. We prove that after any
given number of steps $j$, the expected number of messages at node
$v$ is still $\beta\pi(v)\eta$.  Consider the random walk's
probability transition matrix, say $A$. In this case $Au = u$ for
the vector $u$ having value $\pi(v)$ (since this $\pi(v)$ is the
stationary distribution). Now the number of messages we started with
at any node $i$ is proportional to its stationary distribution,
therefore, in expectation, the number of messages at any node
remains the same.

To calculate $\mathbb{E}[X^j(e)]$, notice that edge $e$ will receive messages from its two end points,
say $x$ and $y$. The number of messages it receives from node $x$ in expectation is exactly
$\beta\pi(x)\eta\times \alpha\min(\frac{1}{d_x}, \frac{\pi_y}{\pi_x d_y}) \leq \eta\cdot \frac{\pi(x)/\deg(x)}{\min_y \pi(y)/\deg(y)}$. The claim follows.
\end{proof}

The high probability analysis follows the same way as the analysis of Lemma~\ref{lem:phase1}.
%
%
\end{proof}


\begin{lemma}
For any $\eta$ and $\lambda$ such that $\eta\lambda\geq 32\sqrt{\ell}(\log n)^3$, Phase~2 finishes in $\tilde O(\frac{\ell}{\lambda})$ rounds with high probability.
\end{lemma}
\begin{proof}(Sketched)
We first prove a result similar to Proposition~\ref{proposition:first and second moment}
\begin{claim} For
any node $x$, node $y$ and $t = O(m^2)$,
\begin{equation}
\e[N_t^x(y)] \le \frac{8 \pi(y) \sqrt{t+1}}{\alpha\min_x \pi(x)/\deg(x)}
\,.
\end{equation}
\end{claim}
\begin{proof} The proof is similar to the proof of Lemma~\ref{proposition:first and second moment} except that  \[c=\alpha\min_x \pi(x)/\deg(x).\]
It follows that
\begin{eqnarray*}
\e[N_t^x(y)] & = & \e_x[  \sum_{i=0}^t \bone_{\{X_i=y\}}] = \sum_{i=0}^t P^i(x,y) \\
& \le &  \frac{4 \pi(y)}{c} \sum_{i=0}^t \frac{1}{\sqrt{i+1}} , \ \ \mbox{ (using the above inequality  (\ref{one_sided_decay})) } \\
& \le & \frac{8 \pi(y) \sqrt{t+1}}{\alpha\min_x \pi(x)/\deg(x)}\,.
\end{eqnarray*}
\end{proof}

By following the rest of the proof of Lemma~\ref{lemma:visits bound}, we conclude the following.
\begin{claim}
For any nodes $x_1, x_2, \ldots, x_k$, and $\ell=O(m^2)$, \[\Pr\bigl(\exists y\ s.t.\
\sum_{i=1}^k N_\ell^{x_i}(y) \geq 32 \frac{\pi(y)}{\alpha\min_x\pi(x)/\deg(x)} \sqrt{k\ell+1}\log
n+k\bigr) \leq 1/n\,.\]
\end{claim}
Following the proof of Lemma~\ref{lem:phase2-probabilistic}, we have that each node $y$ is used as a connector
at most
$$\frac{32 (\frac{\pi(y)}{\alpha\min_x\pi(x)/\deg(x)}) \sqrt{\ell}(\log n)^3}{\lambda}$$
times with probability at least $1-2/n$. Additionally, observe
that we have prepared this many walks in Phase~1; i.e., after
Phase~1, each node $x$ has
$$\eta\cdot \frac{\pi(x)}{\alpha \min_x \frac{\pi(x)}{\deg(x)}} \geq \frac{32 (\frac{\pi(x)}{\alpha\min_y\pi(y)/d(y)}) \sqrt{\ell}(\log n)^3}{\lambda}$$
short walks. The claim follows.
\end{proof}

%


\subsection{$k$ Walks where Sources output Destinations ($k$-RW-SoD)}
\label{sec:k-RW-SoD}

In this section we extend our results to $k$-RW-SoD using the following lemma.

\begin{lemma}
\label{lem:conversion} Given an algorithm that solves $k$-RW-DoS in
$O(S)$ rounds, for any $S$, one can extend the algorithm to solve
$k$-RW-SoD in $O(S+k+D)$ rounds.
\end{lemma}

The idea of the above lemma is to construct a BFS tree and have each
destination node send its ID to the corresponding source via the
root. By using upcast and downcast algorithms~\cite{peleg}, this can
be done in $O(k+D)$ rounds.

\begin{proof}
Let the algorithm that solves $k$-RW-DoS perform one walk each from
source nodes $s_1, s_2, \ldots, s_k$. Let the destinations that
output these sources be $d_1, d_2, \ldots, d_k$ respectively. This
means that for each $1\leq i\leq k$, node $\deg(x)$ has the ID of source
$s_i$. To prove the lemma, we need a way for each $\deg(x)$ to
communicate its own ID to $s_i$ respectively, in $O(k+D)$ rounds.
The simplest way to do this is for each node ID pair $(\deg(x), s_i)$ to
be communicated to some fixed node $r$, and then for $r$ to
communicate this information to the sources $s_i$. This is done by
$r$ constructing a BFS tree rooted at itself. This step takes $O(D)$
rounds. Now, each destination $\deg(x)$ sends its pair $(\deg(x), s_i)$ up
this tree to the root $r$. This can be done in $O(D + k)$ rounds
using an upcast algorithm \cite{peleg}.
 Node $r$ then uses the same BFS tree to route back the pairs to the appropriate sources. This again takes $O(D+k)$ rounds using a downcast algorithm \cite{peleg}.
\end{proof}

Applying Theorem~\ref{thm:kwalks} and Lemma~\ref{lem:conversion}, the following theorem follows.
\begin{theorem}
Given a set of $k$ sources, one can perform $k$-RW-SoD after random
walks of length $\ell$ in $\tilde O(\sqrt{k\ell D}+D+k)$ rounds.
\end{theorem}

\section{Applications}

In this section, we present two applications of our algorithm.

\subsection{A Distributed Algorithm for Random Spanning Tree}
\label{sec:rst}
 We now present an algorithm for generating a random spanning
tree (RST) of an unweighted undirected network in $\tilde{O}(\sqrt{m}D)$ rounds with
high probability.  The approach is to simulate Aldous and Broder's \cite{aldous,Broder89}
RST algorithm  which is as follows. First, pick one arbitrary node as a root. Then, perform a random walk from the root node until
all nodes are visited. For each non-root node, output the edge that
is used for its first visit.  (That is, for each non-root node $v$,
if the first time $v$ is visited is $t$ then we output the edge $(u,v)$
where $u$ is the node visited at time $t-1$.)
The output edges clearly form a spanning tree and this spanning tree
is shown to come from a uniform distribution among all spanning trees of the graph~\cite{aldous,Broder89}.
The running  time of this algorithm is bounded by the time to visit all the nodes
of the
the graph which can shown to be $\tilde{O}(mD)$ (in the worst case, i.e., for any undirected, unweighted graph) by Aleniunas et
al.~\cite{aleliunas}.

This algorithm can be simulated on the distributed network by our
random walk algorithm as follows. The algorithm can be viewed in phases. Initially, we pick a root node
arbitrarily and set $\ell=n$. In each phase, we run $\log n$ (different) walks of length
$\ell$ starting from the root node (this takes
$\tilde{O}(\sqrt{\ell D})$ rounds using our distributed random walk algorithm).  If none of the $O(\log n)$ different walks cover all nodes (this can be easily checked in $O(D)$ time), we
double the value of $\ell$ and start a new phase, i.e., perform again $\log n$  walks of length $\ell$. The algorithm continues
until one walk of length $\ell$ covers all nodes. We then use
such walk to construct a random spanning tree: As the result of this
walk, each node knows its position(s) in the walk (cf. Section~\ref{sec:rw_analysis}), i.e., it has a list
of steps in the walk that it is visited. Therefore, each non-root
node can pick an edge that is used in its first visit by
communicating to its neighbors. Thus at the end of the algorithm,
each node can know which of its adjacent edges belong to the output tree.  (An additional $O(n)$ rounds may be
used to deliver the resulting tree to a particular node if needed.)

We now analyze the number of rounds in term of $\tau$, the expected
cover time of the input graph. The algorithm takes $O(\log\tau)$ phases
before $2\tau\leq \ell\leq 4\tau$, and since one of $\log n$ random
walks of length $2\tau$ will cover the input graph with high
probability, the algorithm will stop with $\ell\leq 4\tau$ with high
probability. Since each phase takes $\tilde{O}(\sqrt{\ell D})$ rounds, the total number of rounds is $\tilde
{O}(\sqrt{\tau D})$ with high probability.
Since  $\tau=\tilde{O}(mD)$, we have the following theorem.

\begin{theorem}
The  algorithm described above generates a uniform random spanning tree
in $\tilde O(\sqrt{m}D)$ rounds with high probability.
\end{theorem}

\subsection{Decentralized Estimation of Mixing Time}
\label{sec:mixingtime}
We now present an algorithm to estimate the mixing time of a graph from a specified source. Throughout this section, we assume that the graph is connected and non-bipartite (the conditions under which mixing time is well-defined).
The main idea in estimating the mixing time is, given a source node, to run many random
walks of length $\ell$ using the approach described in the previous
section, and use these to estimate the distribution induced by the $\ell$-length
random walk. We then compare the distribution at length $\ell$, with
the stationary distribution to determine if they are {\em close}, and if not, double $\ell$ and retry. For this approach, one issue that we need to address is how to compare two distributions with few samples efficiently (a well-studied problem). We introduce some definitions before formalizing our approach and theorem.





\begin{definition}[Distribution vector]
Let $\pi_x(t)$ define the probability distribution vector reached after $t$ steps when the initial distribution starts with probability $1$ at node $x$. Let $\pi$ denote the stationary distribution vector.
\end{definition}

\begin{definition}[$\tau^x(\delta)$ ($\delta$-near mixing time), and $\tau^x_{mix}$ (mixing time) for source $x$]
Define $\tau^x(\delta) = \min t : ||\pi_x(t) - \pi||_1 < \delta$. Define $\tau^x_{mix} = \tau^x(1/2e)$.
\end{definition}

The goal is to estimate $\tau^x_{mix}$. Notice that the definitions
of $\tau^x(\delta)$ and $\tau^x_{mix}$ are consistent due to the following standard
monotonicity property of distributions.

\begin{lemma}\label{lem:monotonicity}
$||\pi_x(t+1) - \pi||_1 \leq  ||\pi_x(t) - \pi||_1$.
\end{lemma}
\begin{proof}
We need to show that the definition of mixing times are consistent, i.e. monotonic in $t$ the number of steps of the random walk. This is folklore but for completeness, we show this via simple linear algebra and the definition of distributions. Let $A$ denote the transpose of the transition probability matrix of the graph being considered. That is, $A(i,j)$ denotes the probability of transitioning from node $j$ to node $i$. Further, let $x$ denote any probability vetor. Now notice that we have $||Ax||_1 \le ||x||_1$; this follows from the fact that the sum of entries of any column of $A$ is $1$ (since it is a Markov chain), and the sum of entries of the vector $x$ is $1$ (since it is a probability distribution vector).

Now let $\pi$ be the stationary distribution of the graph corresponding to $A$. This implies that if $\ell$ is $\delta$-near mixing, then $||A^{\ell}u - \pi||_1 \leq \delta$, by the definition of $\delta$-near mixing time. Now consider $||A^{\ell+1}u - \pi||_1$. This is equal to $||A^{\ell+1}u - A\pi||_1$ since $A\pi = \pi$.  However, this reduces to $||A(A^{\ell}u - \pi)||_1 \leq \delta$ (which again follows from the fact that $A$ is stochastic). It follows that $(\ell+1)$ is $\delta$-near mixing.
\end{proof}





To compare two distributions, we use the technique of Batu et. al.~\cite{BFFKRW} to determine if the distributions are $\delta$-near. Their result (slightly restated) is summarized in the following theorem.


\begin{theorem}[\cite{BFFKRW}]
For any $\epsilon$, given $\tilde{O}(n^{1/2}poly(\epsilon^{-1}))$ samples of a distribution $X$
over $[n]$, and a specified distribution $Y$, there is a test that outputs PASS with high probability if $|X-Y|_1\leq \frac{\epsilon^3}{4\sqrt{n}\log n}$, and outputs FAIL with high probability if $|X-Y|_1\geq 6\epsilon$.
\end{theorem}

The distribution $X$ in our context is some distribution on nodes and $Y$ is the stationary distribution, i.e., $Y(v)=\deg(v)/(2m)$ (recall that $m$ is the number of edges in the network). In this case, the algorithm used in the above theorem can be simulated in a distributed network in $\tilde O(D+2/\log(1+\epsilon))$ rounds, as in the following theorem.

\begin{theorem}\label{thm:batu}
For any $\epsilon$, given $\tilde{O}(n^{1/2}poly(\epsilon^{-1}))$ samples of a distribution $X$
over $[n]$, and a stationary distribution $Y$, there is a $\tilde O(D+2/\log(1+\epsilon))$-time test that outputs PASS with high probability if $|X-Y|_1\leq \frac{\epsilon^3}{4\sqrt{n}\log n}$, and outputs FAIL with high probability if $|X-Y|_1\geq 6\epsilon$.
\end{theorem}
\begin{proof}
We now give a brief description of the algorithm of Batu et. al.~\cite{BFFKRW} to illustrate that it can in fact be simulated on the distributed network efficiently.
The algorithm partitions the set of nodes into $k$ buckets, where $k=(2/\log(1+\epsilon))\log n$, based on $Y$ (the stationary distribution in this case). Denote these buckets by $R_1, \ldots, R_k$. Each bucket $R_i$ consists of all nodes $v$ such that $\frac{(1+\epsilon)^{i-1}}{n\log n}\leq Y(v) <\frac{(1+\epsilon)^{i}}{n\log n}$. Since $n$, $m$ and $\epsilon$ can be broadcasted to all nodes in $O(D)$ rounds and each node $v$ can compute its stationary distribution $Y(v)=\deg(v)/(2m)$, each node can determine which bucket it is in in $O(D)$ rounds.

Now, we sample $\tilde{O}(n^{1/2}poly(\epsilon^{-1}))$ nodes based on distribution $X$. Each of the $\tilde{O}(n^{1/2}poly(\epsilon^{-1}))$ sampled nodes from $X$ falls in one of these buckets. We let $\ell_i$ be the number of sampled nodes in bucket $R_i$ and let $Y(R_i)$ be the distribution of $Y$ on $R_i$. The values of $\ell_i$ and $Y(R_i)$, for all $i$, can compute and sent to some central node in $O(k)=\tilde O(2/\log(1+\epsilon))$ rounds. Finally, the central node uses this information to determine the output of the algorithm. We refer the reader to \cite{BFFKRW} for a precise description.
%
%
\end{proof}

Our algorithm starts with $\ell=1$ and runs $K=\tilde{O}(\sqrt{n}\polylog(\epsilon^{-1}))$ walks (for choice of $\epsilon=1/12e$) of length $\ell$ from the specified source $x$. As the test of comparison with the stationary distribution outputs FAIL, $\ell$ is doubled. This process is repeated to identify the largest $\ell$ such that the test outputs FAIL with high probability and the smallest $\ell$ such that the test outputs PASS with high probability. These give lower and upper bounds on the required $\tau^x_{mix}$ respectively. Our resulting theorem is presented below.

\begin{theorem}\label{thm:mixmain}
Given a graph with diameter $D$, a node $x$ can find, in $\tilde{O}(n^{1/2} + n^{1/4}\sqrt{D\tau^x(\delta)})$ rounds, a time
$\tilde{\tau}^x_{mix}$ such that $\tau^x_{mix}\leq \tilde{\tau}^x_{mix}\leq \tau^x(\delta)$, where $\delta = \frac{1}{6912e\sqrt{n}\log n}$.
\end{theorem}
\begin{proof}
For undirected unweighted graphs, the
stationary distribution of the random walk is known and is
$\frac{deg(i)}{2m}$ for node $i$ with degree $deg(i)$, where $m$ is
the number of edges in the graph.  If a source node in the network knows the degree distribution, we only need
$\tilde{O}(n^{1/2}poly(\epsilon^{-1}))$ samples from a distribution to
compare it to the stationary distribution.  This can be achieved by
running {\sc MultipleRandomWalk} to obtain $K = \tilde{O}(n^{1/2}poly(\epsilon^{-1}))$ random walks. We choose $\epsilon = 1/12e$.
To find the approximate mixing time, we try out
increasing values of $\ell$ that are powers of $2$.  Once we find the
right consecutive powers of $2$, the monotonicity property admits a
binary search to determine the exact value for the specified $\epsilon$.

The result
in~\cite{BFFKRW} can also be adapted to compare with the stationary distribution even if the source does not know the entire distribution. As described previously, the source only needs to know the {\em count} of number of nodes with stationary distribution in given buckets. Specifically, the buckets of interest are at most $\tilde{O}(n^{1/2}poly(\epsilon^{-1}))$ as the count is required only for buckets were a sample is drawn from. Since each node knows its own stationary probability (determined just by its degree), the source can broadcast a specific bucket information and recover, in $O(D)$ steps, the count of number of nodes that fall into this bucket. Using upcast, the source can obtain the bucket count for each of these at most $\tilde{O}(n^{1/2}poly(\epsilon^{-1}))$ buckets in $\tilde{O}(n^{1/2}poly(\epsilon^{-1}) + D)$ rounds.

By Theorem \ref{thm:kwalks}, a source node can obtain $K$ samples from $K$ independent random walks of length $\ell$ in $\tilde{O}(K + \sqrt{K\ell D})$ rounds. Setting $K=\tilde{O}(n^{1/2}poly(\epsilon^{-1}) + D)$ completes the proof.
\end{proof}

Suppose our estimate of $\tau^x_{mix}$ is close to the mixing time of the graph defined as $\tau_{mix} = \max_{x}{\tau^x_{mix}}$, then this would allow us to estimate several related quantities. Given a mixing time $\tau_{mix}$, we can approximate the spectral gap ($1-\lambda_2$) and the conductance ($\Phi$) due to the
known relations that $\frac{1}{1-\lambda_2}\leq \tau_{mix}\leq \frac{\log n}{1-\lambda_2}$ and $\Theta(1-\lambda_2)\leq \Phi\leq \Theta(\sqrt{1-\lambda_2})$ as shown in~\cite{JS89}.

\section{Concluding Remarks}\label{sec:conclusion}

This paper gives a tight upper bound on the time complexity of distributed computation of random walks in undirected networks. Thus the running time of our algorithm is optimal (within a poly-logarithmic factor), matching the lower bound that was shown recently \cite{NanongkaiDP11}.
However, our upper bound for performing $k$ independent random walks may not be tight and it will be interesting to resolve this.

While  the focus  of this paper is on time complexity, message complexity is also important. In particular, our message complexity for computing $k$ independent random walks of length $\ell$ is $\tilde O(m\sqrt{\ell D}+n\sqrt{\ell/D})$ which can be worse than the naive algorithm's $\tilde O(k\ell)$ message complexity.
It would be important to come up with an algorithm that is round efficient and yet has smaller message complexity.
In a subsequent paper \cite{infocom2012}, we have addressed this issue partly and shown that, under certain assumptions, we can extend our algorithms  to be message efficient also.

We presented two algorithmic applications of our distributed random walk algorithm: estimating mixing times and computing random spanning trees. It would be interesting to improve upon these results. For example, is there a $\tilde{O}(\sqrt{\tau^x_{mix}} + n^{1/4})$ round algorithm to estimate $\tau^x$; and is there an algorithm for estimating the mixing time (which is the worst among all
starting points)? Another  open question is whether there exists a $\tilde{O}(n)$ round (or a faster) algorithm for RST?

There are several interesting directions to take this work further. Can these techniques be useful for estimating the second eigenvector of the transition matrix (useful for sparse cuts)? Are there efficient distributed algorithms for random walks in directed graphs (useful for PageRank and related quantities)? Finally, from a practical standpoint, it is important to develop algorithms that are robust to failures and it would be nice to extend our techniques to handle such node/edge failures. This can be useful for  doing decentralized computation
in large-scale dynamic networks.

\bibliographystyle{plain}
\bibliography{Distributed-RW}


\end{document}